\documentclass{amsart}
\usepackage[a4paper,hmargin={2.5cm,2.5cm},vmargin={2.5cm,2.5cm},
heightrounded, marginparwidth=2.2cm, marginparsep=0.1cm]{geometry}

\usepackage{color}

\usepackage[utf8]{inputenc}
\usepackage[T1]{fontenc}

\usepackage[leqno]{amsmath}
\usepackage{amssymb,amsthm}
\usepackage[mathscr]{euscript}
\usepackage{bbm}
\usepackage{dsfont}
\usepackage{stmaryrd}
\usepackage{enumitem}
\usepackage[all,cmtip,2cell]{xy}
\UseAllTwocells
\usepackage{mathtools}
\usepackage{multicol}
\usepackage{etex} 
\usepackage{tikz}

\theoremstyle{plain}
\newtheorem{theorem}{Theorem}[section]
\newtheorem{lemma}[theorem]{Lemma}
\newtheorem{proposition}[theorem]{Proposition}
\newtheorem{corollary}[theorem]{Corollary}

\theoremstyle{definition}
\newtheorem{definition}[theorem]{Definition}
\newtheorem{examples}[theorem]{Examples}
\newtheorem{example}[theorem]{Example}
\newtheorem{assumption}[theorem]{Assumption}

\theoremstyle{remark}
\newtheorem{remark}[theorem]{Remark}

\newlist{tfae}{enumerate}{1}
\setlist[tfae,1]{label = (\roman*)}

\usepackage{chngcntr}
\counterwithin*{equation}{section}

\newcommand{\calC}{\mathcal{C}}
\newcommand{\calD}{\mathcal{D}}

\newcommand{\catfont}[1]{\mathsf{#1}}

\newcommand{\catC}{\catfont{C}}
\newcommand{\catA}{\catfont{A}}
\newcommand{\catB}{\catfont{B}}

\newcommand{\ORD}{\catfont{Ord}}
\newcommand{\PMET}{\catfont{PMet}}
\newcommand{\SET}{\catfont{Set}}
\newcommand{\STONE}{\catfont{Stone}}
\newcommand{\SPECTRAL}{\catfont{Spec}}
\newcommand{\TOP}{\catfont{Top}}

\newcommand{\COMPHAUS}{\catfont{CompHaus}}
\newcommand{\STCOMP}{\catfont{StablyComp}}

\newcommand{\HAUS}{\catfont{Haus}}

\newcommand\adjunctop[2]{\xymatrix@=8ex{\ar@{}[r]|{\bot}\ar@<1mm>@/^2mm/[r]^{{#2}} & \ar@<1mm>@/^2mm/[l]^{{#1}}}}
\newcommand\adjunct[2]{\xymatrix@=8ex{\ar@{}[r]|{\top}\ar@<1mm>@/^2mm/[r]^{{#2}} & \ar@<1mm>@/^2mm/[l]^{{#1}}}}

\newcommand{\field}[1]{\mathds{#1}}

\newcommand{\N}{\field{N}}

\newcommand{\R}{\field{R}}

\DeclareMathOperator{\id}{id}
\DeclareMathOperator{\Id}{Id}
\DeclareMathOperator{\ev}{ev}
\DeclareMathOperator{\im}{im}
\newcommand{\op}{\mathrm{op}}
\newcommand{\df}[1]{\emph{\textbf{#1}}}
\DeclareMathOperator{\upc}{\uparrow\!}
\DeclareMathOperator{\downc}{\downarrow\!}

\newcommand{\Cp}{\mathcal{C}=(C\to X_i)_{i\in I}}
\newcommand{\C}{\mathcal{C}}

\newcommand{\Dp}{\mathcal{D}=(D\to X_i)_{i\in I}}
\newcommand{\POSCH}{\catfont{PosComp}}

\def\lana{ 
  [ \hspace{-0.2em} ( 
}
\def\rana{ 
   ) \hspace{-0.2em} ] 
}
\DeclareMathOperator{\out}{\mathsf{out}}
\DeclareMathOperator{\nxt}{\mathsf{nxt}}
\DeclareMathOperator{\mov}{\mathsf{mov}}
\DeclareMathOperator{\hd}{\mathsf{hd}}
\DeclareMathOperator{\tl}{\mathsf{tl}}
\newcommand{\ie}{\emph{i.e.}}
\newcommand{\eg}{\emph{e.g.}}
\newcommand{\cf}{\emph{cf.}}
\newcommand{\Coalg}{\catfont{CoAlg}}
\newcommand{\MH}{\mathscr{H}}
\newcommand{\Rz}{\mathsf{T}}

\newcommand{\Dur}{\mathsf{D}}

\newcommand{\overU}{\overline{U}}
\newcommand{\overF}{\overline{F}}
\newcommand{\comp}{\mathbin{\boldsymbol{\cdot}}}

\def\pv#1#2{\langle#1 \rangle#2}
\newcommand{\pfs}{\hspace{0.2cm}}
\newcommand{\Vie}{\mathcal{V}}
\usetikzlibrary{patterns}

\title{Limits in categories of Vietoris coalgebras}

\author{Dirk Hofmann}
\address{Center for Research and Development in Mathematics and
  Applications, Department of Mathematics, University of Aveiro,
  3810-193 Aveiro, Portugal}
\email{dirk@ua.pt}

\author{Renato Neves}
\address{INESC TEC (HASLab) \& Universidade do Minho, Portugal}
\email{nevrenato@di.uminho.pt}

\author{Pedro Nora}
\address{Center for Research and Development in Mathematics and
  Applications, Department of Mathematics, University of Aveiro,
  3810-193 Aveiro, Portugal}
\email{a28224@ua.pt}


\date{\today}
\usepackage{hyperref} 

\begin{document}

\begin{abstract}
  Motivated by the need to reason about hybrid systems, we study
  limits in categories of coalgebras whose underlying functor is a
  Vietoris polynomial one --- intuitively, the topological analogue of
  a Kripke polynomial functor. Among other results, we prove that
  every Vietoris polynomial functor admits a final coalgebra if it
  respects certain conditions concerning separation axioms and
  compactness. When the functor is restricted to some of the
  categories induced by these conditions the resulting categories of
  coalgebras are even complete.

  As a practical application, we use these developments in the
  specification and analysis of non-deterministic hybrid systems, in
  particular to obtain suitable notions of stability, and behaviour.
\end{abstract}

\maketitle

\section{Introduction}

\subsection{Motivation and context}

Coalgebras \cite{rutten2000,Ada05,introcoalg} form a powerful theory
of state-based transition systems where definitions and results are
formulated at a high level of genericity that covers several families
of systems at once, from deterministic automata and Kripke frames to
different kinds of probabilistic models. Traditionally, these
formulations are elaborated in a set-based context; \ie\ no further
structure in the system's state space than that of a set is assumed.
In many cases, however, a switch of context is needed. The projects on
the coalgebraic foundations of stochastic systems, where the Giry
functor and measurable spaces have a central role (\cf\
\cite{viglizzo:2005,prakash_2009,Doberkat:2009}), are evident examples
of this.  Research on coalgebras over Stone spaces (\eg\
\cite{Kupke:2004,Venema10,venema2014}) and coalgebras over
pseudometric spaces \cite{BaldanBKK14} forms equally important
cases. In \cite{Kupke:2004,Venema10,venema2014}, the aim is to provide
a suitable coalgebraic semantics for finitary modal logics by taking
advantage of a Vietoris functor, while in \cite{BaldanBKK14} is to
introduce a notion of distance between states.

In this paper our focus is on coalgebras over arbitrary
\emph{topological spaces}, because we believe that they provide
important mechanisms to the design and analysis of \emph{hybrid
  systems} \cite{tabuada09,alur2015,Stauner_01}. Briefly put, hybrid
systems are those that possess both discrete and continuous behaviour,
a result of the complex interaction between digital devices, and
physical processes like velocity, movement, temperature, and time. Two
recurring examples are the cruise control system, basically a digital
device with influence over velocity, and the \emph{bouncing ball}. In
the latter, movement and velocity have a continuous nature, while the
impact on the ground is assumed to be a discrete event that
instantaneously alters the current velocity.  As we will see in the
following sections, such an interaction between discrete and
continuous behaviour calls for a shift from the set-based setting to
richer contexts, in particular to topological ones so that suitable
notions of stability, bisimulation, and behaviour can be
obtained. These are the practical motivations for the theoretical
results that this paper provides. But we stress that coalgebras over
topological spaces have the potential for much more -- the works
\cite{viglizzo:2005,prakash_2009,Doberkat:2009,Venema10,venema2014,
  BaldanBKK14,Kupke:2004}, for example, elegantly attest this. Our
results are therefore applicable to a much broader context than that
of hybrid systems.

Each functor $F : \catC \to \catC$ induces a category of coalgebras
$\Coalg(F)$ that can be seen as a framework for a particular family of
state-based transition systems, whose transition type is determined by
$F : \catC \to \catC$ (\cf\ \cite{rutten2000}). The powerset
$\mathscr{P} : \SET \to \SET$, for example, often associated with
non-deterministic behaviour, gives rise to Kripke frames.

In such a context, the systematic study of (co)limits in categories of
coalgebras is a natural research line. In fact, final coalgebras,
which form a specific type of limit, are often searched for, as they
encode a canonical notion of behaviour for all $F$-coalgebras. Another
special kind of limit, equalisers of coalgebras, is extensively used
in coalgebraic specification (\cf\ \cite{rutten2000,Ada05}). It
provides a notion of subsystem, and is essential to characterise a
system induced by a set of coequations.

\subsection{Contributions and related work}

As mentioned before, this paper concerns coalgebras over arbitrary
topological spaces. More concretely, coalgebras whose underlying
functor is defined over the category $\TOP$ of topological spaces and
continuous maps. Analogously to what has already been done in $\SET$
(\eg\ \cite{rutten2000,gumm2001}), the aim here is to investigate the
existence of limits in categories of coalgebras whose underlying
functor is \df{Vietoris polynomial} --- the topological analogue of a
Kripke polynomial functor. The former is called `Vietoris polynomial'
because it arises from the composition of different Vietoris functors
\cite{Vie22,Mic51,CT97} (the topological analogues of the powerset
functor) with polynomial functors over $\TOP$. To keep the
nomenclature simple, we call every coalgebra whose underlying functor
is Vietoris polynomial a \df{Vietoris coalgebra}.

As composites of constant, (co)product, identity, and powerset
functors, Kripke polynomial functors have long since been recognised
as a particularly relevant class of functors (\cf\
\cite{rutten2000,alex09,Kupke:2004}).  They are intuitive and the
corresponding coalgebras subsume several types of state-based
systems. Moreover, they are well-behaved in regard to the existence of
limits in their categories of coalgebras if the powerset functor is
submitted to certain cardinality restrictions.  We will see that
somewhat similar results hold for Vietoris polynomial functors as
well.  Actually, an instance of a Vietoris functor, which we call
\df{compact Vietoris functor}, has already been studied multiple times
in the coalgebraic setting (\eg\
\cite{Kupke:2004,Venema10,venema2014,Fred16}), and will appear in a
book on coalgebras that is currently in preparation \cite{AMM16_tmp}.
In particular, \cite{Kupke:2004} shows that compact Vietoris
polynomial functors in the category $\STONE$ of Stone spaces and
continuous maps admit a final coalgebra. Also, document \cite{Fred16}
presents a theorem that can be generalised to show that the compact
Vietoris functor in the category $\COMPHAUS$ of compact Hausdorff
spaces and continuous maps, admits a final coalgebra.  In fact, this
generalised result is also implicitly mentioned in \cite[page
245]{Eng89}. Related to this, but in a broader setting, we collect a
number of results scattered in coalgebraic and topological literature,
and

\begin{itemize}

\item add to this collection some results of our own. In particular,
  we generalise Hughes' theorem (Theorem \ref{theo_hughes}) and prove
  that, under certain conditions, functors between categories of
  coalgebras are \emph{topological} \cite{Ada05}. Topological functors
  have powerful properties such as the existence of left and right
  adjoints, lifting of limits, and lifting of factorisations.

\item This collection of results allows us to obtain several new
  results about limits in categories of Vietoris coalgebras.  For
  example, that categories of polynomial coalgebras over $\TOP$ are
  complete, and that categories of compact Vietoris coalgebras over
  $\COMPHAUS$ are complete as well. Using in particular \cite[Lemma
  B]{Zen70}, we also show that categories of compact Vietoris
  coalgebras are complete in the category $\HAUS$ of Hausdorff spaces
  and continuous maps. Moreover we will see that all categories of
  Vietoris coalgebras over $\TOP$ have equalisers.

\item We then take advantage of the limit-preserving properties of the
  inclusion functors $\COMPHAUS \to \TOP$ and $\HAUS \to \TOP$ to show
  that every compact Vietoris polynomial functor $F: \TOP \to \TOP$
  that can be restricted either to $\COMPHAUS$ or $\HAUS$ admits a
  final coalgebra.

\end{itemize}

Our setting is a broader one also because we consider different
instances of Vietoris functors, a particular case being what we call
the \df{lower Vietoris functor}, studied in a coalgebraic setting in
\cite{BKR07}.

\begin{itemize}
\item We will show that every lower Vietoris polynomial functor
  behaves well in the category $\STCOMP$ of stably compact spaces and
  spectral maps. In particular, that its category of coalgebras is
  complete.

\item In order to extend these results to more variants of Vietoris
  functors, we study the existence of adjunctions between categories
  of coalgebras. One positive result is that, assuming the existence
  of a monomorphic natural transformation between the underlying
  functors, such an adjunction exists under mild conditions.
\end{itemize}

To illustrate the practical side of these developments, and, more
generally, the potential of coalgebras over $\TOP$ to the design and
analysis of hybrid systems, we argue that the coalgebraic
specification in $\SET$ of the bouncing ball has some deficiencies.
Among them, the incapability to reason about the system's
stability, and the non-existence of a suitable final coalgebra if
non-determinism is taken into account.  We will see that these issues
can be solved, to some extent, by adopting the category $\TOP$ as the
underlying semantic universe.

\subsection{Roadmap} The ensuing section introduces some categorial
notions, provides an overview, and extends some results about limits
in categories of coalgebras. Then, it formally reviews the concept of
Vietoris coalgebra and different instances of Vietoris functors --- as
already mentioned, our agenda has a broader scope than most
coalgebraic literature on Vietoris functors, which mainly focuses on
one specific case.

Section \ref{Sec_limits} starts with our study about polynomial
coalgebras over $\TOP$, and topological functors between categories of
coalgebras. Then, it adds two instances of Vietoris functors (the
lower and the compact) to the mix which, as expected, introduce a
number of difficulties. A number of topological concepts are recalled
at this point to help us achieve some of the results mentioned above.

Section \ref{Sec_morelimits} explores the existence of adjunctions
between categories of coalgebras induced by natural transformations
relating functors on the underlying categories. As already stated,
this allows to extend the results of the previous section to
subfunctors of Vietoris polynomial ones, thus covering at once several
variants of Vietoris functors.

Section \ref{Sec_work} illustrates an application of this work to the
design of hybrid systems.  Finally, Section \ref{Sec_conc} suggests
possible research lines for future work and concludes.

We assume that the reader has basic knowledge of category theory
\cite{Mac71,AHS90}, topology \cite{Kelley,JGL-topology}, and
coalgebras \cite{rutten2000,Ada05,introcoalg}.

\section{Preliminaries}
\label{Sec_pre}
\subsection{Categorial notions}

\noindent
Some categorial notions that the reader may not frequently meet will
be used.  This section provides a brief overview about them.

\noindent

\begin{definition}
  A diagram $D:I\to\catC$ is said to be \df{codirected} whenever $I$
  is a codirected partially ordered set, that is, $I$ is non-empty and
  for all $i,j\in I$ there is some $k\in I$ with $k\to i$ and
  $k\to j$. A cone for a codirected diagram is called a \df{codirected
    cone}. In particular, a limit of a codirected diagram is called
  \df{codirected}.
\end{definition}

\begin{example}
  Inverse sequence (or $\omega^\op$) diagrams, which have the shape
  depicted below, are codirected.
  \[
    \cdot \longleftarrow \cdot \longleftarrow \cdot \longleftarrow
    \dots
  \]
  Inverse sequence diagrams have a central role in showing that a
  given functor admits a final coalgebra (see Theorem \ref{Th:omega}).
\end{example}

\begin{remark}
  \label{rem:codirect_posch}
  The codirected limit of a diagram $D:I\to\SET$ is given by the
  subset
  \[
    \left\{(x_i)_{i\in I}\in \prod_{i\in I}D(i)\mid \forall j\to i \in
      I, D(j\to i)(x_j)=x_i\right\}
  \]
  of the product $\prod_{i\in I}D(i)$.
\end{remark}

\begin{definition}
  A category $\catC$ is said to be \df{connected} if it is non-empty
  and every two objects $A,B \in \catC$ can be connected by a finite
  zig-zag of morphisms as depicted below.
  \begin{flalign*}
    A \leftarrow \cdot \rightarrow \dots \leftarrow \cdot \rightarrow
    B
  \end{flalign*}
  A diagram $D : I \to \catC$ is called \df{connected diagram} if $I$
  is connected, and a limit of $D$ is called \df{connected limit} if
  $D : I \to \catC$ is connected.
\end{definition}

\begin{examples}
  Equalisers and codirected limits are two examples of connected
  limits.
\end{examples}

We will see in the following section that polynomial functors over
$\TOP$ preserve connected limits, in particular codirected ones.

\begin{definition}
  Let $F:\catA\to \catB$ be a functor. A cone $\Cp$ in $\catA$ is said
  to be \df{initial with respect to} $F$ if for every cone $\Dp$ and
  every morphism $h:FD\to FC$ such that
  $F\mathcal{D}=F\mathcal{C}\comp h$, there exists a unique
  $\catA$-morphism $\bar{h}:D\to C$ with
  $\mathcal{D}=\mathcal{C} \comp \bar{h}$ and $h=F\bar{h}$.
\end{definition}

We simply say that the cone is \df{initial} whenever no ambiguities
arise. 

\begin{examples}
  \phantomsection
  \label{exs:initial}
  \begin{enumerate}
  \item \label{wrtTOP} A cone $(f_i:X\to X_i)$ in $\TOP$ is initial
    with respect to the forgetful functor $\TOP\to\SET$ if and only if
    $X$ is equipped with the so called initial (weak)
    topology. Explicitly, the topology generated by the subbasis
    \[
      f_i^{-1}(U)\hspace{2em}(i\in I, U \subseteq X_i \text{ open}).
    \]
  \item In the category $\COMPHAUS$ of compact Hausdorff spaces and
    continuous maps, a monocone is initial in $\TOP$ (\cf\
    \cite[Theorem 4.4.27]{JGL-topology}). Interestingly, the converse
    also holds, as a initial cone in $\TOP$ whose domain is a T$_0$
    space is necessarily mono.
  \end{enumerate}
\end{examples}

\begin{remark}\label{rem:initial_1}
  In Example~\ref{exs:initial}(\ref{wrtTOP}) the subbasis is actually
  a basis if the cone is codirected.
\end{remark}

\begin{theorem}[{\cite[Proposition 13.15]{AHS90}}]
  \label{theo_gen_init}
  Let $F:\catA\to\catB$ be a limit preserving faithful functor and
  $D:I\to \catA$ a diagram. A cone $\C$ for $D$ is a limit of $D$ if
  and only if the cone $F \mathcal{C}$ is a limit of $FD$ and
  $\mathcal{C}$ is initial with respect to $F$.
\end{theorem}

\subsection{Limits in categories of coalgebras}

\noindent
Let $F : \catC \to \catC$ be an arbitrary functor. Then, dually to the
algebraic case, one can easily show that colimits in $\Coalg(F)$ exist
if they do so in $\catC$ (\cf\ \cite{rutten2000,Ada05}).  The story
about limits in categories of coalgebras is, however, more complex. In
this subsection we review some well-known results on this topic, a
special focus being given to those more relevant to the paper. We
start at a generic level, with the following two theorems (\cf\
\cite{rutten2000,Ada05}).

\begin{theorem}
  \label{Th:omega}
  Let $\catC$ be a category with a final object $1$ and
  $F:\catC\to\catC$ a functor. If the category $\catC$ has a limit $L$
  of the diagram
  \[
    1\longleftarrow F1\longleftarrow FF1\longleftarrow \dots
  \]
  and $F$ preserves this limit, then the canonical isomorphism
  $L \to F L$ is a final $F$-coalgebra.
\end{theorem}

\begin{theorem}
  \label{Th:limpres}
  Assume that $F : \catC \to \catC$ preserves limits of a certain
  type. Then the forgetful functor $\Coalg(F) \to \catC$ creates
  limits of the same type.
\end{theorem}

\noindent
An important consequence of the last theorem is that $\Coalg(F)$ has
all types of limit that $\catC$ has and that the functor
$F: \catC \to \catC$ preserves. Unfortunately, as we will witness
later, this assumption is often too strong. Resorting to the notion of
covarietor, the following results will be more helpful.

\begin{definition} 
  A functor $F:\catC\to\catC$ is said to be a \df{covarietor} if the
  canonical forgetful functor $\Coalg(F)\to\catC$ is left adjoint.
\end{definition}

\noindent
This adjoint situation allows to take advantage of the theory of
(co)monads regarding (co)completeness of Eilenberg-Moore (co)algebras
to derive the following theorem (\cf\ \cite{Lin69}).

\begin{theorem}
  \label{theo_linton}
  Let $F$ be a covarietor over a complete category. If $\Coalg(F)$ has
  equalisers then $\Coalg(F)$ is complete.
\end{theorem}

\noindent
Related to this, Hughes proved the following theorem in \cite[Theorem
2.4.2]{hughes2001}.

\begin{theorem}
  \label{theo_hughes}
  Let $\catC$ be regularly wellpowered, cocomplete, and possess
  equalisers. Moreover, assume that it has an (Epi,
  RegMono)-factorisation structure, and that the functor
  $F : \catC \to \catC$ preserves regular monomorphisms. Then
  $\Coalg(F)$ has equalisers.
\end{theorem}

\noindent
Using Theorem \ref{theo_linton}, one can then easily deduce the
following corollary.

\begin{corollary}
  If the conditions in the last theorem hold and, additionally,
  $\catC$ is complete and $F$ is a covarietor, then the category
  $\Coalg(F)$ is complete.
\end{corollary}

\noindent
We refer the interested reader to other results on limits in
categories of coalgebras. In particular, the work of Kurz
\cite{kurz2001}, which shows that $\Coalg(F)$ is complete whenever it
has a suitable factorisation structure, $F$ is a covarietor, and
$\catC$ is complete; document \cite{gumm2001}, where the authors study
the existence of equalisers and products in categories of coalgebras
over $\SET$; and the documents \cite{powerw98,gumm2001}, where the
existence of limits is studied under the assumption of $F$ being
bounded.

To close this section, we provide an improvement to Hughes'
theorem. We start with notation.
\begin{definition}
  For a small category $I$, a \df{cone for $I$} in a category $\catC$
  is given by a functor $D:I\to\catC$ together with a cone
  $(X\to D(i))_{i\in I}$ for $D$. Given a class $\mathcal{M}$ of cones
  for $I$, the category $\catC$ is called
  \df{$\mathcal{M}$-wellpowered} if for every functor $D:I\to\catC$
  there is up to isomorphism only a set of cones for $D$ in
  $\mathcal{M}$.
\end{definition}

Our first lemma is in the spirit of \cite[Section 12]{AHS90} and shows
that ``cocompleteness almost implies completeness''.

\begin{lemma}\label{lem:complete}
  Let $\catC$ be a cocomplete category and $I$ a small
  category. Furthermore, let $E$ be a class of $\catC$-morphisms and
  $\mathcal{M}$ be a class of cones for $I$ in $\catC$. If $\catC$ is
  $\mathcal{M}$-wellpowered and every cone for $I$ has a
  $(E,\mathcal{M})$-factorisation, then $\catC$ has limits of shape
  $I$.
\end{lemma}
\begin{proof}
  We will show that the diagonal functor
  \[
    \Delta:\catC\to\catC^I
  \]
  has a right adjoint, using Freyd's General Adjoint Functor Theorem
  (see \cite{Mac71}).  By assumption, $\catC$ is cocomplete and the
  functor $\Delta$ clearly preserves colimits, so we just need to show
  that the Solution Set Condition holds. In this context it unfolds to
  the following condition: for every functor $D:I\to\catC$, there is a
  set $\mathcal{S}$ of cones for $D$ such that every cone
  $(f_i:C\to D(i))_{i\in I}$ for $D$ factors through a cone in
  $\mathcal{S}$.

  Since $\catC$ is $\mathcal{M}$-wellpowered we have, by assumption, a
  set $\mathcal{S}$ of representants for $D$ in
  $\mathcal{M}$. Moreover $\catC$ has a
  $(E, \mathcal{M})$-factorisation system for $I$, which means that
  the cone $(f_i:C\to D(i))_{i\in I}$ can be factorised as depicted
  below
  \[
    \xymatrix{C\ar[rr]^{f_i}\ar[dr]_e && D(i)\\ & A\ar[ru]_{g_i}}
  \]
  with the cone $(g_i : A \to D(i))_{i \in I}$ in $\mathcal{S}$.
\end{proof}

The factorisation system assumed in this lemma may appear to be rather
unconventional, but, as the following remarks will show, it actually
emerges from mild conditions.

\begin{remark}\label{rem:arrows-to-cones}
  Consider a category $\catC$ equipped with classes $E$ and $M$ of
  morphisms so that every morphism in $\catC$ has a
  $(E,M)$-factorisation and $\catC$ is $M$-wellpowered. Under
  additional assumptions, such factorisations can be extended to cones
  for $I$. To be more concrete:
  \begin{enumerate}
  \item Assume that $\catC$ has products. Then we put
    \[
      \mathcal{M}= \left \{\text{all cones $(f_i:X\to D(i))_{i\in I}$
          for $I$ where
          $\langle f_i\rangle_{i\in I}:X\to\prod_{i\in I}D(i)$ is in
          $M$} \right \}.
    \]
    Clearly, every cone for $I$ is $(E,\mathcal{M})$-factorisable (see
    \cite[Proposition 15.19]{AHS90}), and $\catC$ is
    $\mathcal{M}$-wellpowered.
  \item In order to relate the previous lemma with Hughes' theorem,
    assume that $I=\{1\rightrightarrows 2\}$ and that $E$ is contained
    in the class of epimorphisms of $\catC$. The class of cones
    \[
      \mathcal{M}= \left \{\text{all cones $(f_i:X\to D(i))_{i\in I}$
          for $I$ with $f_1$ in $M$} \right \},
    \]
    makes every cone for $I$ $(E,\mathcal{M})$-factorisable and the
    category $\catC$ is $\mathcal{M}$-wellpowered.
  \end{enumerate}
\end{remark}

Finally, we apply the results above to categories of coalgebras.

\begin{theorem}\label{cor:general-completeness-coalgebras}
  Let $F:\catC\to\catC$ be an endofunctor over a cocomplete category
  $\catC$ and let $I$ be a small category. If $\catC$ is
  $(E,\mathcal{M})$-structured for cones for $I$,
  $\mathcal{M}$-wellpowered and $F$ sends cones in $\mathcal{M}$ to
  cones in $\mathcal{M}$, then $\Coalg(F)$ has limits of shape $I$.
\end{theorem}
\begin{proof}
  The assumptions guarantee that the factorisation system in $\catC$
  lifts to the category $\Coalg(F)$ (\cf\ \cite{Ada05,Che14}).  The
  claim then follows from Lemma~\ref{lem:complete}.
\end{proof}

Let us now relate in a more precise manner the previous theorem with
Hughes' theorem.

\begin{theorem}\label{cor:general-completeness-coalgebras2}
  Let $F:\catC\to\catC$ be an endofunctor over a cocomplete category
  $\catC$. If $\catC$ is regularly well-powered, has an (Epi,
  RegMono)-factorisation structure and $F : \catC \to \catC$ preserves
  regular monomorphisms, then $\Coalg(F)$ has equalisers.
\end{theorem}

\begin{proof}
  Let $I = \{ 1 \rightrightarrows 2 \}$ and use Remark
  \ref{rem:arrows-to-cones}(2) to provide a
  $(E,\mathcal{M})$-factorisation system for cones for $I$.  The
  category $\catC$ is clearly $\mathcal{M}$-wellpowered and by a
  simple reasoning one shows that $F$ sends cones in $\mathcal{M}$ to
  cones in $\mathcal{M}$. Now apply Theorem
  \ref{cor:general-completeness-coalgebras}.
\end{proof}

The last result shows that Hughes' assumption of $\catC$ having
equalisers is not necessary. Another interesting point is the ability
that we gain to reason not just about equalisers but any type of
limit. We will take advantage of this generalisation in the next
section (see Corollary~\ref{cor:Vietoris-pol-preserves-mono-cones}).

Note also that the following corollaries can be obtained almost for
free.

\begin{corollary}
  Let $F: \SET \to \SET$ be a functor that preserves monocones of a
  certain type.  Then the category $\Coalg(F)$ has limits of the same
  type.
\end{corollary}

Recall that $\TOP$ is an (Epi,initial monocones)-category and an
(RegEpi,monocones)-category (\cf\ \cite[Examples
15.3~(6)]{AHS90}). The following result can then be derived.

\begin{corollary}
  Let $F: \TOP \to \TOP$ be a functor that preserves either small
  monocones or small initial monocones of a certain type. Then the
  category $\Coalg(F)$ has limits of the same type.
\end{corollary}

\subsection{Vietoris polynominal functors}

\noindent
Although traditionally considered in $\SET$ (\eg\
\cite{alex09,introcoalg}), the notion of a polynomial functor can be
formally defined at a more generic level.

\begin{definition}
  \label{defn:poly}
  Let $\catC$ be a category with (co)products. We call a functor
  $F : \catC \to \catC$ \df{polynomial} if it can be recursively
  defined from the grammar below
  \begin{flalign*}
    F ::= F + F \mid F \times F \mid A \mid \Id
  \end{flalign*}
  where $A$ corresponds to an object of $\catC$.
\end{definition}

\begin{remark}\label{rem:pol_lim}
  Alternatively, one can define the class of polynomial functors as
  the smallest class of functors $F:\catC\to\catC$ that contains the
  identity functor, all constant functors, and is closed under
  products and sums of functors. Here, for functors
  $F,G:\catC\to\catC$, the product of $F$ and $G$, and the sum of $F$
  and $G$ are, respectively, the composites
  \[
    \catC\xrightarrow{\,\langle
      F,G\rangle\,}\catC\times\catC\xrightarrow{\,\times\,}\catC,
    \hspace{0.2cm} \text{and} \hspace{0.2cm}
    \catC\xrightarrow{\,\langle
      F,G\rangle\,}\catC\times\catC\xrightarrow{\,+\,}\catC.
  \]
  \noindent
  Note that if the functors $F,G : \catC \to \catC$ preserve limits of
  a certain type the functor $F \times G : \catC \to \catC$ preserves
  limits of the same type as well. Note also that
\end{remark}

\begin{proposition}
  The functor $(+) : \TOP \times \TOP \to \TOP$ preserves connected
  limits.
\end{proposition}

\begin{proof}
  It is well-known that the functor $(+) : \SET \times \SET \to \SET$
  preserves connected limits. Then observe that
  $(+) : \TOP \times \TOP \to \TOP$ preserves initial cones and apply
  Theorem \ref{theo_gen_init}.
\end{proof}

\begin{corollary}\label{pol_conn_lim}
  If the functors $F,G : \TOP \to \TOP$ preserve connected limits the
  functor $F + G : \TOP \to \TOP$ preserves connected limits as well.
\end{corollary}

\noindent
In the set-based context, the powerset functor
$\mathscr{P} : \SET \to \SET$ is traditionally used in conjunction
with polynomial functors to bring non-deterministic behaviour into the
scene, the resulting functor being a so called \df{Kripke polynomial
  functor}. The situation is more complex in the topological context
because a number of functors can be seen as `analogues' of the
powerset. Most of them have their roots in the Hausdorff metric
(\cf~\cite{Pom05,Hau14}) and in Vietoris' ``Bereiche zweiter Ordnung''
\cite{Vie22}. Informally, we call them \df{Vietoris functors}.  The
remainder of this section provides some details about them.

Consider a compact Hausdorff space $X$, the \df{classic Vietoris
  space} $\Vie X$ \cite{Vie22} consists of the set of all closed
subsets of $X$, \ie
\begin{flalign*}
  \Vie X = \{K \subseteq X \mid K\ \text{is closed}\}
\end{flalign*}
\noindent
equipped with the `hit-and-miss topology' generated by the subbasis of
sets of the form
\begin{flalign*}
  & U^\Diamond = \{ A \in \Vie X \mid A \cap U \not = \varnothing \}
  \qquad (``A\ \text{hits}\ U")\>, \\
  & U^\Box = \{ A \in \Vie X \mid A \subseteq U \} \qquad (``A\
  \text{misses}\ X\setminus U"),
\end{flalign*}
where $U \subseteq X$ is open. Nowadays there are several well-studied
variants of this archetype that give rise to endofunctors over
specific subcategories of $\TOP$. The interested reader will find in
\cite{Mic51} and \cite{CT97} more details about these
constructions. For now, we concentrate on two particular cases,
described below.

\begin{examples}
  \phantomsection
  \label{exs:Vietoris-functors}
  \begin{enumerate}
  \item\label{item:compact} For a topological space $X$, define
    $\Vie X=\{ K \subseteq X \mid K \text{ is compact } \}$ with the
    topology generated by the sets $U^\Box$ and $U^\Diamond$, with $U$
    ranging over all open subsets $U\subseteq X$. Then, given a
    continuous map $f : X \rightarrow Y$, define
    $\Vie f : \Vie X \rightarrow \Vie Y$ as $ \Vie f (A) = f [A]$. We
    call this variant \df{compact Vietoris functor}. It is well-known
    that $\Vie X$ is compact Hausdorff whenever $X$ is. In fact, for
    compact Hausdorff spaces this construction coincides with the
    classic one \cite{Vie22}.
  \item\label{item:closed} For a topological space $X$, define
    $\Vie X=\{ K \subseteq X \> | \> K \ \text{is closed} \}$ with the
    topology generated by the sets $U^\Diamond$, with $U$ ranging over
    all open subsets $U\subseteq X$. Then, given a continuous map
    $f : X \rightarrow Y$, define $\Vie f : \Vie X \rightarrow \Vie Y$
    as $\Vie f (A) = \overline{f [A]}$, where $\overline{f [A]}$
    denotes the closure of $f[A]$. This variant is called \df{lower
      Vietoris functor}.
  \end{enumerate}
\end{examples}

\begin{remark} 
  The classic Vietoris construction, with closed sets, does not define
  an obvious functor on $\TOP$. That is, adding the sets $U^\Box$ to
  the subbasis of
  Example~\ref{exs:Vietoris-functors}~(\ref{item:closed}) does not
  define a functor. To see why, consider the set $\{1,2,3\}$ equipped
  with the topology generated by the sets $\{1,2\}$ and $\{2,3\}$. For
  the subspace embedding $i:\{1,2\}\to
  \{1,2,3\}$,~$(Vi)^{-1} [\{1,2\}^\Box]=\{\varnothing,
  \{1\}\}$. However, every open set of $\Vie \{1,2\}$ that contains
  $\{1\}$ contains $\{1,2\}$.
\end{remark}

\noindent
A number of projects on (coalgebraic) modal logic studied the compact
Vietoris functor in the category of Stone spaces
(\eg~\cite{Kupke:2004,venema2014}) and in the category of compact
Hausdorff spaces \cite{BBH12}. The second case was explored by
\cite{CLP91,Pet96,BKR07} in the context of Priestley spaces.

\begin{definition}
  Let $\Vie : \TOP \to \TOP$ be the lower Vietoris functor.  We call a
  functor $F : \TOP \to \TOP$ \df{lower Vietoris polynomial} if it can
  be recursively defined from the grammar below.
  \begin{flalign*}
    F ::= F + F \mid F \times F \mid A \mid \Id \mid \Vie
  \end{flalign*}
  Similarly, if we consider the compact Vietoris functor
  $\Vie : \TOP \to \TOP$ \emph{in lieu} of the lower one, then we
  speak of a \df{compact Vietoris polynomial} functor.
\end{definition}

\section{On limits in categories of Vietoris coalgebras}
\label{Sec_limits}
\subsection{Polynomial functors in $\TOP$}
\label{Subsec_Pol_func}
Using standard results, we now show that for a polynomial functor
$F : \TOP \to \TOP$ the associated category of coalgebras $\Coalg(F)$
is complete. A useful fact for this proof is that the category $\TOP$
is (co)complete (\cf\ \cite{AHS90}). Moreover, note that

\begin{theorem}
  \label{theo_conn}
  All polynomial functors $F : \TOP \to \TOP$ preserve connected
  limits.
\end{theorem}

\begin{proof}
  Clearly the identity functor $\Id : \TOP \to \TOP$ preserves all
  limits, and the constant functor $A : \TOP \to \TOP$ trivially
  preserves connected limits. The claim now follows from Remark
  \ref{rem:pol_lim} and Corollary \ref{pol_conn_lim}.
\end{proof}

\noindent
From the theorem above one can derive the following results in a
straightforward manner.

\begin{proposition}
  \label{prop_regmono}
  All polynomial functors $F: \TOP \to \TOP$ preserve regular
  monomorphisms.
\end{proposition}

\begin{proof}
  First note that the diagrams associated with equalisers are
  connected. Then, recall that a regular monomorphism is an equaliser
  of a pair of morphisms.
\end{proof}

\begin{theorem}
  \label{theo_pres_omega}
  All polynomial functors $F : \TOP \to \TOP$ are covarietors.
\end{theorem}

\begin{proof}
  Since a polynomial functor $F: \TOP \to \TOP$ preserves connected
  limits (Theorem \ref{theo_conn}) it preserves the codirected ones as
  well.  The claim is then a direct consequence of \cite[Theorem
  2.1]{Bar93}.
\end{proof}

\noindent
In regard to equalisers in $\Coalg(F)$, one can easily show that the
necessary requirements to apply Theorem \ref{theo_hughes} are met.
Actually, it is well-known that the category $\TOP$ is regularly
wellpowered (\cf\ \cite{AHS90}), and we already saw that it is
(co)complete. Moreover, it has an (Epi, RegMono)-factorisation
structure (\cf\ \cite{AHS90}).  Therefore,

\begin{corollary}
  \label{cor_eq}
  If $F : \TOP \to \TOP$ is a polynomial functor, the category
  $\Coalg(F)$ has equalisers.
\end{corollary}

\begin{proof}
  A direct consequence of Theorem \ref{theo_hughes} and Proposition
  \ref{prop_regmono}.
\end{proof}

\begin{theorem}
  \label{theo_comp}
  If $F : \TOP \to \TOP$ is a polynomial functor, the category
  $\Coalg(F)$ is complete.
\end{theorem}

\begin{proof}
  Observe that $F$ is a covarietor (Theorem \ref{theo_pres_omega}),
  and that the category $\Coalg(F)$ has equalisers (Corollary
  \ref{cor_eq}). Then, apply Theorem \ref{theo_linton}.
\end{proof}

We will now use `less standard' results to go further than the
previous theorem. More concretely, we will show that not only is
$\Coalg(F)$ complete but also that there is a functor with powerful
properties from $\Coalg(F)$ to the analogous category of coalgebras
over $\SET$. 
By going further we also mean that the results that we will introduce
next may be used in categories different than $\TOP$, prime examples
are the category of preordered sets $\ORD$ and the category of
pseudometric spaces $\PMET$.

The general idea is that starting with a category $\catB$ with good
properties and assuming the existence of a functor $\catA \to \catB$
that lifts these properties to a category $\catA$, there will often be
a functor $\Coalg(\overF) \to \Coalg(F)$ with the same lifting
properties than $\catA \to \catB$ for functors
$\overF : \catA \to \catA$, $F : \catB \to \catB$ making the diagram
below commute.
\[
  \xymatrix{
    \catA \ar[r]^{\overF} \ar[d]_{U} & \catA \ar[d]^{U} \\
    \catB \ar[r]_{F} & \catB }
\]

The following definition recalls the notion of topological functor,
which lifts several properties of a category.
\begin{definition}
  A functor $U : \catA \to \catB$ is called \df{topological} if every
  cone $\calC = (X \to U X_i)_{i \in I}$ in $\catB$ has a $U$-initial
  lifting, \ie\ a initial cone $\calD = (A \to X_i)_{i \in I}$ with
  respect to $U : \catA \to \catB$ such that $\calC = U \calD$.
\end{definition}

\begin{remark}
  Every topological functor is both left and right adjoint, lifts
  limits and certain types of factorisations (see \cite{Ada05}).
\end{remark}

\begin{proposition}
  Consider two categories $\catA, \catB$ a functor
  $U : \catA \to \catB$, endofunctors $\overline{F} : \catA \to \catA$,
  $F : \catB \to \catB$, and a natural transformation
  \begin{flalign*}
    \delta : U \overline{F} \to F U.
  \end{flalign*}
  \noindent
  Then, there is a functor $\overU : \Coalg(F) \to \Coalg(\overF)$
  defined by the equations
  \begin{flalign*}
    & \overU (X,c) = (U X, \delta_X \comp U c), \hspace{1cm} \overU f
    = U f
  \end{flalign*}
  that makes the diagram below commute.
  \[
    \xymatrix{
      \Coalg(\overF) \ar[r] \ar[d]_{\overU} & \catA \ar[d]^{U} \\
      \Coalg(F) \ar[r] & \catB }
  \]
\end{proposition}

\noindent
Moreover,
\begin{proposition}
  If the functor $U : \catA \to \catB$ is faithful and the natural
  transformation $\delta : U \overF \to F U$ is mono, then the induced
  functor $\overU : \Coalg(\overF) \to \Coalg(F)$ is faithful.
\end{proposition}

\begin{proof}
  Direct consequence of the natural transformation
  $\delta : U \overF \to F U$ being mono and the functor
  $U : \catA \to \catB$ being faithful.
\end{proof}


\begin{lemma} \label{lem:induces_top} Assume that the natural
  transformation $\delta : \overF U \to U F$ is mono and $U$ is
  faithful. Let $(f_i : (X,c) \to (Y_i, d_i))_{i \in I}$ be a cone in
  $\Coalg(\overF)$, and $(f_i : X \to Y_i)_{i \in I}$ be initial with
  respect to $U : \catA \to \catB$. Then, the cone
  $(f_i : (X,c) \to (Y_i, d_i))_{i \in I}$ is initial with respect to
  the functor $\overU : \Coalg(\overF) \to \Coalg(F)$.
\end{lemma}

\begin{proof}
  Let $(f_i : (X,c) \to (Y_i, d_i))_{i \in I}$ be a cone in
  $\Coalg(\overF)$ and $(f_i : X \to Y_i)_{i \in I}$ be initial with
  respect to $U : \catA \to \catB$. Then, consider another cone
  $(g_i : (Z,e) \to (Y_i, d_i))_{i \in I}$ in $\Coalg(\overF)$ and
  assume that its $\overline{U}$-image is factorised as shown by the
  diagram below.
  \[
    \xymatrix{ \overU (Z,e) \ar[d]_{h} \ar[dr]^{\overU g_i} &
      \\
      \overU (X,c) \ar[r]_{\overU f_i} & \overU (Y_i, d_i) }
  \]
  \noindent
  The forgetful functor $\Coalg(F) \to \catB$ yields the following
  factorisation of the cone $(U g_i : U Z \to U Y_i)_{i \in I}$.
  \[
    \xymatrix{ U Z \ar[d]_{h} \ar[dr]^{U g_i} &
      \\
      U X \ar[r]_{U f_i} & U Y_i }
  \]
  Since the cone $(f_i : X \to Y_i)_{i \in I}$ is initial with respect
  to $U : \catA \to \catB$, there is a unique arrow
  $\overline{h} : Z \to X$ in $\catA$ such that for all $i \in I$ we
  have
  \begin{flalign*}
    g_i = f_i \comp \overline{h}, \hspace{1cm} U \overline{h} = h.
  \end{flalign*}
  \noindent
  It remains to show that the arrow $\overline{h} : Z \to X$ is also a
  coalgebra homomorphism $\overline{h} : (Z,e) \to (X,c)$.  For this,
  consider the diagram below.
  \[
    \xymatrix{ Z \ar[r]^{\overline{h}} \ar[d]_{e} &
      X \ar[r]^{f_i} \ar[d]^{c} & Y_i \ar[d]^{d_i} \\
      \overF Z \ar[r]_{\overF \; \overline{h}} & \overF X
      \ar[r]_{\overF f_i} & \overF Y_i }
  \]
  \noindent
  By assumption, the equation
  $F h \comp \delta_Z \comp U e = \delta_X \comp U c \comp h$ holds.
  Then reason in the following manner.
  \begin{flalign*}
    & F h \comp \delta_Z \comp U e = \delta_X \comp U c \comp h  & \\
    \equiv \pfs & F U \overline{h} \comp \delta_Z \comp U e = \delta_X
    \comp U c \comp U \overline{h} & \\
    \equiv \pfs & \delta_X \comp U \overF \; \overline{h} \comp U e =
    \delta_X \comp U c
    \comp U \overline{h} & \\
    \Rightarrow \pfs & U \overF \; \overline{h} \comp U e =
    U c \comp U \overline{h} & \\
    \equiv \pfs & U (\overF \; \overline{h} \comp e) =
    U (c \comp \overline{h}) & \\
    \Rightarrow \pfs & \overF \; \overline{h} \comp e = c \comp
    \overline{h}
  \end{flalign*}
\end{proof}

\begin{theorem}
  Assume that $\overF : \catA \to \catA$ preserves initial cones and
  that  $U \overF = F U$. Then if the functor $U : \catA \to \catB$ is topological, the
  functor $\overU : \Coalg(\overF) \to \Coalg(F)$ is topological as
  well.
\end{theorem}

\begin{proof}
  Let $(f_i : (X,c) \to \overU (Y_i,d_i))_{i \in I}$ be a cone in
  $\Coalg(F)$. Since the functor $U : \catA \to \catB$ is topological,
  the induced cone $(f_i : X \to U Y_i)_{i \in I}$ admits a
  $U$-initial lifting
  \begin{flalign*}
    (\overline{f}_i : A \to Y_i)_{i \in I}.
  \end{flalign*}
  By assumption, the cone
  $(\overF \; \overline{f}_i : \overF A \to \overF Y_i)_{i \in I}$ is
  also initial. Moreover, note that the following equations hold,
  \begin{flalign*}
    U \left (A \stackrel{\overline{f}_i}{\to} Y_i \stackrel{d_i}{\to}
      \overline{F} Y_i \right ) & = \left ( X \stackrel{f_i}{\to} U
      Y_i \stackrel{U d_i}{\to} F U Y_i \right ) \\
    U \left (\overline{F} A \stackrel{\overline{F} \;
        \overline{f}_i}{\to} \overline{F} Y_i \right ) & = \left (F X
      \stackrel{F f_i}{\to} F U Y_i \right ) \hspace{1.5cm} (i \in I)
  \end{flalign*}
  and that we have the factorisation below.
  \[
    \xymatrix{
      X \ar[dr]^{U d_i \comp f_i} \ar[d]_{c} & \\
      F X \ar[r]_{F f_i} & F U Y_i }
  \]
  \noindent
  This provides an arrow $\overline{c} : A \to \overF A $ such that
  $U \overline{c} = c$, and that makes the diagram below to commute.
  \[
    \xymatrix{
      A \ar[dr]^{d_i \comp \overline{f}_i} \ar[d]_{\overline{c}} & \\
      \overF A \ar[r]_{\overF\, \overline{f}_i} & \overF Y_i }
  \]
  We thus have a cone
  $(\overline{f}_i : (A,\overline{c}) \to (Y_i,d_i))_{i \in I}$ in
  $\Coalg(\overF)$. To finish the proof recall that the cone
  $(\overline{f}_i : A \to Y_i)_{i \in I}$ is initial with respect to
  the functor $U : \catA \to \catB$ and apply Lemma
  \ref{lem:induces_top}.
\end{proof}

\begin{corollary}
  \label{cor_top_comp}
  Let $U : \catA \to \catB$ be a topological functor and consider two
  functors $\overF : \catA \to \catA$, $F : \catB \to \catB$ such that
  $\overF : \catA \to \catA$ preserves initial cones. Moreover assume
  that $U \overF = F U$.
  Then the category $\Coalg(\overF)$ is complete iff $\Coalg(F)$ is
  complete.
\end{corollary}

The forgetful functor $\TOP \to \SET$ is topological (\cf\
\cite{Ada05}) and it is straightforward to show that all polynomial
functors over $\TOP$ preserve initial cones. Using the previous
corollary this entails that all categories of coalgebras of a
polynomial functor over $\TOP$ are complete.

As hinted before, Corollary \ref{cor_top_comp} has stronger
consequences than Theorem \ref{theo_comp}: it considers all functors
in $\TOP$ that preserve initial cones (and not just the polynomial
ones) and it does not make any assumption about the category $\catA$
being $\TOP$. In fact, the only assumption about the category $\catA$
is that it has a topological functor $\catA \to \catB$. We invite the
reader to examine in \cite{Ada05} several examples of such categories.

\subsection{Some notes about Vietoris functors}

The last corollary is a positive result of our study of limits in
categories of polynomial coalgebras. On the other hand, the addition
of Vietoris functors to the mix brings a whole new level of difficulty
that calls for a number of topological concepts, an investigation of
Vietoris functors and some of their preservation properties. The study
of such properties is the main goal of this section.

\begin{lemma}
  Let $X$ be a topological space and $\mathcal{B}$ a base for the
  topology of $X$.
  \begin{enumerate}
  \item The set $\{B^\Diamond\mid B\in\mathcal{B}\}$ is a subbase for
    the lower Vietoris space $\Vie X$ (\cf\
    Example~\ref{exs:Vietoris-functors}(\ref{item:closed})).
  \item If $\mathcal{B}$ is closed under finite unions, then the set
    $\{B^\Diamond\mid B\in\mathcal{B}\}\cup\{B^\Box\mid
    B\in\mathcal{B}\}$ is a subbase for the compact Vietoris space
    $\Vie X$ (\cf\
    Example~\ref{exs:Vietoris-functors}(\ref{item:compact})).
  \end{enumerate}
\end{lemma}
\begin{proof}
  Let $\mathcal{S}$ be a set of open subsets of $X$. First note that,
  for both the lower and the compact Vietoris space,
  \[
    \left(\bigcup\mathcal{S}\right)^\Diamond=\bigcup \left
      \{S^\Diamond\mid S\in\mathcal{S} \right \}.
  \]
  This proves the first statement. To see that the second one is also
  true, observe that
  \[
    \left(\bigcup\mathcal{S}\right)^\Box=\bigcup\left\{\left(\bigcup\mathcal{F}\right)^\Box\mid
      \mathcal{F}\subseteq\mathcal{S}\text{ finite}\right\}
  \]
  since we only consider compact subsets of $X$.
\end{proof}

\begin{lemma}\label{lem:init}
  Both the compact and the lower Vietoris functor $\Vie :\TOP\to\TOP$
  preserve initial codirected cones.
\end{lemma}

\begin{proof}
  Let $(f_i:X\to X_i)_{i \in I}$ be an initial codirected cone in
  $\TOP$. Then the set
  \[
    \left \{f^{-1}_i(U) \mid i\in I, U\subseteq X_i\text{ open} \right
    \}
  \]
  is a base for the topology of $X$ (Remark
  \ref{rem:initial_1}). Moreover, the base is closed under finite
  unions. Therefore, by the lemma above, the proof follows from the
  equations
  \begin{flalign*}
    ((f_i)^{-1}(U))^\Box = (\Vie f_i)^{-1} (U^\Box) \hspace{2cm}
    ((f_i)^{-1}(U))^\Diamond = (\Vie f_i)^{-1} \left ( U^\Diamond
    \right ),
  \end{flalign*}
  for all $i\in I$ and $U\subseteq X_i$ open, which are
  straightforward to show.
\end{proof}

\begin{theorem}\label{cor:Vietoris-preserves-mono-cones}
  The lower Vietoris functor preserves initial codirected
  monocones. The compact Vietoris functor preserves initial codirected
  monocones of Hausdorff spaces.
\end{theorem}
\begin{proof}
  First note that for a topological space $X$ the lower Vietoris space
  $\Vie X$ is $T_0$, and if $X$ is Hausdorff the compact Vietoris
  space $\Vie X$ is Hausdorff as well (\cf\ \cite{Mic51}).
  Then recall that a initial cone in $\TOP$ whose
  domain is $T_0$ (or $T_2$) is necessarily mono and apply
  Lemma~\ref{lem:init}.
\end{proof}

Together with Proposition~\ref{prop_regmono} it follows:

\begin{corollary}\label{cor:Vietoris-pol-preserves-mono-cones}
  Every compact polynomial functor and every lower polynomial functor
  $F:\TOP\to\TOP$ preserves regular monomorphisms.
\end{corollary}

\begin{proof}
  We already saw that all polynomial functors preserve regular
  monomorphisms (Proposition~\ref{prop_regmono}), and that the lower
  Vietoris functor preserves them as well
  (Theorem~\ref{cor:Vietoris-preserves-mono-cones}). Moreover, we saw
  that the compact Vietoris functor preserves initial monomorphisms
  (Lemma~\ref{lem:init}) and it is straightforward to show that it
  preserves monomorphisms.
\end{proof}

From Theorem~\ref{cor:Vietoris-preserves-mono-cones} and
Corollary~\ref{cor:general-completeness-coalgebras} we obtain the
following results.
\begin{corollary}
  For every lower Vietoris polynomial functor $F: \TOP \to \TOP$ the
  category $\Coalg(F)$ has codirected limits. For every compact
  Vietoris polynomial functor $F: \TOP \to \TOP$ the category
  $\Coalg(F)$ has codirected limits of Hausdorff spaces.
\end{corollary}

\begin{corollary}
  For every Vietoris polynomial functor $F: \TOP \to \TOP$ the
  category $\Coalg(F)$ has equalisers.
\end{corollary}

\begin{proof}
  Direct consequence of Theorem
  \ref{cor:general-completeness-coalgebras2} and Corollary
  \ref{cor:Vietoris-pol-preserves-mono-cones}.
\end{proof}

\begin{remark}
  The assumption above about codirectedness is essential: neither the
  compact nor the lower Vietoris functor $\Vie:\TOP \to \TOP$ preserve
  monocones in general. Take, for instance, a compact Hausdorff space
  $X$ with at least two elements. Then $A=\{(x,x)\mid x\in X\}$ is a
  closed subset of $X\times X$, and $A$ is different from
  $B=X\times X$. However, with $\pi_1:X\times X\to X$ and
  $\pi_2:X\times X\to X$ denoting the projection maps,
  \[
    \Vie \pi_1(A)= \Vie \pi_1(B) = X = \Vie \pi_2(A) = \Vie \pi_2(B);
  \]
  which shows that the cone
  $(\Vie \pi_1:\Vie(X\times X)\to \Vie X,\Vie \pi_2: \Vie (X\times
  X)\to \Vie X)$ is not mono.
\end{remark}

\noindent
Theorem~\ref{cor:Vietoris-preserves-mono-cones} shows some good
behaviour with respect to codirected initial monocones. However, none
of the functors of Examples~\ref{exs:Vietoris-functors} preserves
codirected limits in $\TOP$.

\begin{examples}
  \phantomsection
  \label{exs:Vietoris-vs-codirected-limits}
  \begin{enumerate}
  \item We consider $I=\N$ with the natural order, and the functor
    $D:\N\to\SET$ which sends $n\le m$ to the inclusion map
    $\{0,\dots n\}\hookrightarrow\{0,\dots,m\}$. Clearly, the set of
    natural numbers $\N$ is a colimit of this directed diagram. Then,
    the composite $\SET(-,\N) \comp D^\op :\N^\op\to\SET$ yields a
    codirected diagram with limit $\SET(\N,\N)$, the limit projections
    $p_n:\SET(\N,\N)\to\SET(D(n),\N)$ being given by
    restriction. Equipping all sets with the indiscrete topology, we
    obtain a codirected limit in $\TOP$. The compact Vietoris functor
    does not send this limit to a monocone since $(Vp_n)_{n\in\N}$
    cannot distinguish between the sets $\SET(\N,\N)$ and
    \[
      \{f:\N\to\N\mid \{n\in\N\mid f(n)\neq 0\}\text{ is finite}\}.
    \]
  \item The next example is based on the ``empty inverse limit'' of
    \cite{Wat72}. Here $I$ is the set of all finite subsets of $\R$,
    with order being containment $\supseteq$. For $F\in I$, let $D(F)$
    be the discrete space of all injective functions $F\to\N$, and the
    map $D(G\supseteq F)$ is given by restriction. Note that each
    connecting map $D(G\supseteq F)$ is surjective. Then the limit of
    this diagram in $\TOP$ is empty since an element of this limit
    would define an injective function $\R\to\N$. The lower Vietoris
    functor sends the limit cone for $D$ to a monocone but not to a
    limit cone since the limit of $\Vie D$ has at least two elements:
    $(\varnothing)_{F\in I}$ and $(D(F))_{F\in I}$. Using the
    indiscrete topology instead of the discrete one shows that the
    lower Vietoris functor does not preserve codirected limits of
    diagrams of compact spaces and closed maps.
  \item In the example above we can use other topologies to show that
    the lower or the compact Vietoris functor does not preserve
    certain codirected limits. As an example, we consider here $\N$
    equipped with the topology
    \[
      \{\upc n\mid n\in\N\}\cup\{\varnothing\};
    \]
    where $\upc n=\{k\in\N\mid n\le k\}$. Note that $\N$ is T$_0$ and
    every non-empty collection of open subsets of $\N$ has a largest
    element with respect to inclusion $\subseteq$. The latter implies
    that, for every finite set $F$, every subset of $\N^F$ is compact.
    To see this, let $C\subseteq\N^F$ and assume that $C$ is covered
    by basic open subsets of $\N^F$:
    \[
      C\subseteq \bigcup_{i \in I} \uparrow i_1 \times \dots \times
      \uparrow i_n.
    \]
    We already know that for every $k \in F$ the family
    $(\uparrow j_k)_{j \in I}$ contains a largest element with respect
    to inclusion. This allows us to construct a finite subcover in the
    following manner: for each $k \in F$ let $S_k$ be an element of
    the cover whose $k$-projection is the largest element of the
    family $(\uparrow j_k)_{j \in I}$.  Then,
    \[
      C\subseteq \bigcup_{k \in F} S_k
    \]
    and the family $(S_k)_{k \in F}$ is a finite subcover.  We
    conclude that $C$ is compact.

    With $I$ being as in the previous example, we consider now $D(F)$
    as a subspace of $\N^F$. Then, for every $G\supseteq F$, the map
    $D(G\supseteq F):D(G)\to D(F)$ is continuous. Hence, this
    construction defines a codirected diagram $D:I\to\TOP$ where each
    $D(F)$ is T$_0$, compact, and locally compact; and the limit of
    this diagram is empty. With the same argument as above, neither
    the lower nor the compact Vietoris functor preserve this limit.
  \end{enumerate}
\end{examples}

\subsection{Vietoris polynomial functors}
\label{sec:via-codirected-limits}

Section~\ref{Subsec_Pol_func} studied limits in categories of
polynomial coalgebras, essentially by analysing the preservation of
connected limits in $\TOP$ and by providing sufficient conditions for
the existence of topological functors between categories of
coalgebras.  In the current section our focus is on Vietoris
coalgebras. In fact, Examples~\ref{exs:Vietoris-vs-codirected-limits}
already showed that it is highly problematic to consider all
topological spaces, because the lower and the compact Vietoris
functors do not preserve codirected limits in $\TOP$.  Hence, we will
restrict our attention to different subcategories of $\TOP$ where more
positive results appear.
  

\begin{definition}
  A topological space $X$ is called \df{stably compact} whenever $X$
  is $T_0$, locally compact, well-filtered and every finite
  intersection of compact saturated subsets is compact~\cite{Jun04}.
  A continuous map between stably compact spaces is called
  \df{spectral} whenever the inverse image of compact saturated
  subsets is compact. Stably compact spaces and spectral maps form a
  category which we denote by $\STCOMP$.
\end{definition}

\begin{remark}
  Note that every stably compact space is compact. More information on
  this type of space can be found in \cite{GHK+03} and \cite{Jun04}.
\end{remark}

\begin{theorem}\label{thm:propriedades_STCOMP_1}
  The category $\STCOMP$ is complete and regularly wellpowered. The
  inclusion functor $\STCOMP\to\TOP$ preserves limits and finite
  coproducts.
\end{theorem}
\begin{proof}
  It is straightforward to check that the finite coproduct of stably
  compact spaces is stably compact (\cf\ \cite[Proposition
  9.2.1]{JGL-topology}). The other claims follow from monadicity of
  $\STCOMP\to\TOP$ which is shown in \cite{Sim82}. We note that
  \cite{Sim82} uses the designation \emph{well-compacted} instead of
  stably compact.
\end{proof}

Further properties of $\STCOMP$ can be easily derived if ones uses a
order-theoretic perspective. 

\begin{definition}
  A \df{partially orderered compact space} is a triple $(X,\le,\tau)$
  consisting of a set $X$, a partial order $\le$ on $X$ and a compact
  topology $\tau$ on $X$ so that the set
\[
  \{(x,y)\in X\times X\mid x\le y\}
\]
is closed with respect to the product topology.
\end{definition}

\begin{remark}
  Every partially ordered compact space $(X,\le,\tau)$ is necessarily
  Hausdorff as the antisymmetry property of the relation $\le$ implies
  that the diagonal $\{(x,x)\mid x\in X\}$ is closed in $X\times X$.
\end{remark}

The category $\STCOMP$ is isomorphic to the category $\POSCH$ of
partially orderered compact spaces and monotone continuous maps (\cf\
\cite{GHK+80}). The isomorphism $\POSCH\to\STCOMP$ commutes with the
underlying forgetful functors to $\SET$, sending a partially ordered
compact space $(X,\leq,\tau)$ to the stably compact space with the
same underlying set and the topology defined by the upper-open sets of
$(X,\leq,\tau)$. Its inverse functor uses the \df{specialisation
  order} of a topological space, defined by
$x \leq y \iff x \in \overline{\{y\}}$. It maps a stably compact space
$(X,\tau)$ into a space $(X,\tau',\leq)$ where the relation $\leq$ is
the specialisation ordering and $\tau'$ the \df{patch topology} of
$(X,\tau)$, \ie\ the topology generated by the complements of compact
saturated subsets and also the opens in $(X,\tau)$.

\begin{remark}
 The canonical forgetful functor $\POSCH\to\COMPHAUS$ has a
left adjoint which equips a compact Hausdorff space with the discrete
order. Using the isomorphism above, the adjunction
\[
  \POSCH\adjunct{\text{discrete}}{\text{forgetful}}\COMPHAUS
\]
reads in the language of stably compact spaces as
\[
  \STCOMP\adjunct{\text{inclusion}}{\text{patch}}\COMPHAUS.
\]
\end{remark}

In the sequel we will freely jump between both perspectives.

\begin{theorem}\label{thm:propriedades_STCOMP_2}
  The category $\POSCH$ is cocomplete and the epimorphisms of $\POSCH$
  are precisely the surjective morphisms.
\end{theorem}
\begin{proof}
  Cocompleteness of $\POSCH$ follows from \cite[Corollary
  2]{Tho09}. Combining several results of \cite{Nac65}, it is shown in
  \cite{HN16_tmp} that every epimorphism in $\POSCH$ is surjective.
\end{proof}

Clearly, (Surjections,Substructure) is a factorisation structure for
morphisms in $\POSCH$. Since the surjections are precisely the
epimorphisms in $\POSCH$, we conclude that $\POSCH$ is
(Epi,RegMono)-structured, and thus also the category
$\STCOMP$. Moreover, the regular monomorphisms in $\STCOMP$ are
precisely the topological subspace embeddings.


Let us turn our attention back to the study of Vietoris functors with
the isomorphism $\STCOMP\simeq\POSCH$ in mind. The lower Vietoris
functor on $\TOP$ restricts to a functor $\Vie:\STCOMP\to\STCOMP$
(\cf\ \cite{Sch93}). Its counterpart on $\POSCH$ can be described in
the following manner.

\begin{proposition}
  Under the isomorphism $\STCOMP\simeq\POSCH$, the lower Vietoris
  functor $\Vie:\STCOMP\to\STCOMP$ corresponds to the functor
  \[
    \POSCH\to\POSCH
  \]
  which sends a partially ordered compact space $X$ to the space of
  all lower-closed subsets of $X$, with order inclusion
  $\subseteq$, and compact topology generated by the sets
  \begin{align}\label{eq:1}
    & \{A\subseteq X\mid A\text{ lower-closed and }A\cap U\neq\varnothing\}\hspace{1em}\text{($U\subseteq X$ upper-open)},  \\
    & \{A\subseteq X\mid A\text{ lower-closed and } A\cap K=\varnothing\}\hspace{1em}\text{($K\subseteq X$ upper-closed)}.\notag
  \end{align}
  Given a map $f:X\to Y$ in $\POSCH$, the functor returns the map that
  sends a lower-closed subset $A\subseteq X$ to the down-closure
  $\downc f[A]$ of $f[A]$.
\end{proposition}
\begin{proof}
  Let $(X,\le,\tau)$ be a partially ordered compact space with
  corresponding stably compact space $(X,\sigma)$. Clearly, the
  underlying set of $\Vie(X,\sigma)$ is the set of all lower-closed
  subsets of $X$. We will show that the patch topology of
  $\Vie(X,\sigma)$ coincides with the topology defined by
  \eqref{eq:1}. First note that every set of the form
  \[
    \{A\subseteq X\mid A\text{ lower-closed and }A\cap
    U\neq\varnothing\}\hspace{1em}\text{($U\subseteq X$ upper-open)},
  \]
  is open in $\Vie(X,\sigma)$ and therefore is also in the patch
  topology. For $K\subseteq X$ upper-closed, the complement of the set
  \[
    \{A\subseteq X\mid A\text{ lower-closed and } A\cap K=\varnothing\}
  \]
  is equal to $K^\Diamond$. Using Alexander's Subbase Theorem, it is
  straightforwad to verify that $K^\Diamond$ is compact in
  $\Vie(X,\sigma)$. Since the specialisation order of $\Vie(X,\sigma)$
  is subset inclusion, $K^\Diamond$ is also saturated. Hence, the
  topology defined by \eqref{eq:1} is coarser than the patch topology
  of $\Vie(X,\sigma)$. Since it is also Hausdorff, by \cite[Lemma
  2.2]{Jun04}, both topologies coincide (\cf\ \cite{Eng89}). In
  particular, the construction of the proposition defines indeed a
  partially ordered compact space.

  In regard to maps in $\POSCH$, \cite[Proposition 4 on page
  44]{Nac65} tells that for every map $f:X\to Y$ in $\POSCH$ and every
  lower-closed subset $A\subseteq X$, the down-closure $\downc f[A]$
  of $f[A]$ is closed in $Y$, and therefore coincides with the closure
  of $f[A]$ in the stably compact topology of $Y$.
\end{proof}

Recall that the lower Vietoris functor preserves codirected initial
monocones (see
Theorem~\ref{cor:Vietoris-preserves-mono-cones}). Hence, for every
codirected diagram $D:I\to\STCOMP$ with limit cone
$(p_i:L_D\to D(i))_{i\in I}$, the canonical comparison map
\[
  h:\Vie L_D\to L_{\Vie D},\; K\mapsto (\overline{p_i[K]})_{i\in I}
\]
is an embedding. To show that $\Vie :\STCOMP\to\STCOMP$
preserves these limits, we are left with the task of proving that $h$
is also surjective. To do so, we use the fact that $\STCOMP$ inherits
a nice characterisation of codirected limits from the category
$\COMPHAUS$. A first hint of the latter characterisation is in
\cite{Bou42}, but, to the best of our knowledge, is rarely used in the
literature. Actually, we were not able to find a proof in the
literature, except for \cite{Hof99}; so we sketch a proof below.

\begin{theorem}\label{thm:codirected-limits-COMPHAUS}
  Let $D:I\to\COMPHAUS$ be a codirected diagram and
  $\C=(p_i:L\to D(i))_{i\in I}$ a cone for $D$. The following
  conditions are equivalent:
  \begin{enumerate}
  \item The cone $\C$ is a limit of $D$.
  \item \label{item:1} The cone $\C$ is mono and, for every $i\in I$,
    the image of $p_i$ contains the intersection of the images of all
    $D(j\to i)$, in symbols
    \[
      \im p_i\supseteq\bigcap_{j{\to}i}\im D(j\to i).
    \]
  \end{enumerate}
\end{theorem}

\begin{proof}
  Assume first that $(p_i:L\to D(i))_{i\in I}$ satisfies the two
  conditions and let $(f_i:X\to D(i))_{i\in I}$ be a cone for $D$. Let
  $x\in X$, and, for every $i\in I$, put
  $A_i=p_i^{-1}(f_i(x))$. Clearly, $A_i$ is closed, moreover, $A_i$ is
  non-empty since
  \[
    \im f_i\subseteq \bigcap_{j{\to}i}\im D(j\to i) = \im p_i
  \]
  Since the family $(A_i)_{i\in I}$ is codirected and $L$ is compact,
  there is some $z\in\bigcap_{i\in I}A_i$. We put $f(x)=z$, this way
  we define a map $f:X\to L$ with $p_i\cdot f=f_i$, for all $i\in
  I$. Since $(p_i:L\to D(i))_{i\in I}$ is a monocone, we conclude that
  $(p_i:L\to D(i))_{i\in I}$ is a limit of $D$. Conversely, if
  $(p_i:L\to D(i))_{i\in I}$ is a limit, then it is clearly a
  monocone. Let now $i_0\in I$ and
  $x\in \bigcap_{j{\to}i_0}\im D(j\to i_0)$. We may assume that $i_0$
  is final in $I$. For each $i\in I$, we put
  \[
    A_i=\{(x_i)_{i\in I}\in\prod_{i\in I}D(i)\mid x_{i_0}=x\text{ and,
      for all } i\to j\in I, x_j=D(i\to j)(x_i)\}.
  \]
  Then $A_i$ is non-empty, and it is a closed subset of
  $\prod_{i\in I}D(i)$ since it is an equaliser of continuous maps
  between Hausdorff spaces. Furthermore, for $i\to j\in I$,
  $A_i\subseteq A_j$. Hence there is some $z\in\bigcap_{i\in I}A_i$;
  by construction, $z\in L$ and $p_{i_0}(z)=x$.
\end{proof}

\begin{remark}
  For every cone $(p_i:C\to D(i))_{i\in I}$ the inequality
  $\im p_i\subseteq \bigcap_{j{\to}i}\im D(j\to i)$ holds. Hence, in
  the theorem above, the reverse inequality, distinguishes monocones
  from limit cones.
\end{remark}





\begin{proposition}
  \label{prop:up_commutes_intersection}
  Let $\mathcal{A}$ be a codirected set of closed subsets of a
  partially ordered compact space $X$. Then,
  $\downc \bigcap_{A\in\mathcal{A}}A=\bigcap_{A\in\mathcal{A}}\downc A$.
\end{proposition}

\begin{proof}
  Clearly,
  $\downc
  \bigcap_{A\in\mathcal{A}}A\subseteq\bigcap_{A\in\mathcal{A}}\downc
  A$. To show that the reverse inequality holds, consider
  $z \in \bigcap_{A\in\mathcal{A}}\downc A$. Then, for every
  $A\in\mathcal{A}$, the set $\upc z \cap A$ is non-empty, and closed
  because $\{z\}$ is compact (\cf \cite[Proposition 4 on page
  44]{Nac65}). Morevover, since $\mathcal{A}$ is codirected, the set
  $\{\upc z \cap A\ |\ A\in \mathcal{A}\}$ has the finite intersection
  property. Therefore, by compactness, it follows that
  $\upc z \cap \bigcap_{A\in \mathcal{A}} A \neq \varnothing$, which
  implies that $z \in \downc \bigcap_{A\in\mathcal{A}}A$.
\end{proof}

\begin{proposition}
  Let $D:I\to\POSCH$ be a codirected diagram,
  $(p_i:L_D\to D(i))_{i\in I}$ a limit for $D$ and
  $(L_{\Vie D}\to \Vie D(i))_{i\in I}$ a limit for $\Vie D:I\to\POSCH$. Then the 
  function $h:\Vie L_D\to L_{\Vie D}$ defined by
  $K\mapsto (\downc p_i[K])_{i\in I}$ is surjective.
\end{proposition}

\begin{proof}
  Let $(K_i)_{i\in I} \in L_{\Vie D}$. For every $i\in I$,
  $K_i\subseteq D(i)$ is closed, hence, $K_i \in \POSCH$. For every
  $i\in I$ and $j\to i\in I$, take $K(i)$ as $K_i$ and $K(j\to i)$ as
  the continuous and monotone map of type $K_j\to K_i$ given by the
  restriction of $D(j\to i)$ to $K_j$.  This way, by
  Remark~\ref{rem:codirect_posch}, we obtain a codirected diagram
  $K:I\to\POSCH$ such that for every $j\to i\in I$,
  $\downc K(j\to i)[K(j)]=[K(i)]$.

  Let $(p_i:L_K\to K(i))_{i\in I}$ be a limit for $K$. By construction,
  $L_K\subseteq L_D$ is lower-closed. Thus, $L_K\in \Vie L_D$. We claim
  that $h(L_k)=(K_i)_{i\in I}$.  Let $i_0\in I$. Since the following
  diagram of forgetful functors
  \[
    \xymatrix@R=15pt@C=0.5pt{
      \POSCH \ar[dr] \ar[rr] & & \COMPHAUS \ar[dl] \\
      & \SET & }
  \]
  commutes and the functor $\POSCH\to\COMPHAUS$ preserves limits, from
  Theorem~\ref{thm:codirected-limits-COMPHAUS} we
  obtain
  \[
    p_{i_0}[L_K]=\bigcap_{j\to i_0} K(j\to i_0)[K_j].
  \]
  Therefore, by Propostion~\ref{prop:up_commutes_intersection},
  \[
    \downc p_{i_0}[L_K]=\downc \bigcap_{j\to i_0} K(j\to i_0)[K(j)]=
    \bigcap_{j\to i_0} \downc K(j\to i_0)[K(j)]=K_{i_0}.
  \]
\end{proof}

As expected, we obtain the following results.

\begin{corollary}
  \label{cor:proper-preserve-codirected-limits}
  The lower Vietoris functor $\Vie :\STCOMP\to\STCOMP$ preserves
  codirected limits.
\end{corollary}

\begin{proof}
  A direct consequence of the previous proposition.
\end{proof}

\begin{corollary}
  \label{cor:theo_codl_lower}
  All lower Vietoris polynomial functors $F:\STCOMP\to\STCOMP$
  preserve codirected limits.
\end{corollary}

\begin{proof}
  Analogous to that of Theorem~\ref{theo_conn}.
\end{proof}

\begin{theorem}
  \label{theo_lower_comp}
  For every lower Vietoris polynomial functor $F:\STCOMP\to\STCOMP$,
  the category $\Coalg(F)$ is complete.
\end{theorem}

\begin{proof}
  Firstly, observe that
  Theorems~\ref{thm:propriedades_STCOMP_1},~\ref{thm:propriedades_STCOMP_2}
  and Corollary~\ref{cor:Vietoris-pol-preserves-mono-cones} guarantee
  the hypothesis of Theorem~\ref{theo_hughes}, therefore the category
  $\Coalg(F)$ has equalisers. Then the assertion follows from
  Corollary \ref{cor:theo_codl_lower} and \cite[Theorem 2.1]{Bar93}.
\end{proof}

In regard to final coalgebras, there is still room to improve the
theorem above. Indeed, the inclusion functor $I : \STCOMP \to \TOP$ is
well-behaved with respect to limits, in particular it preserves and
reflects them (\cf~\cite{Sim82}); this allows us to derive the
following theorem.

\begin{theorem}
  Every lower Vietoris polynomial functor in $\TOP$ that can be
  restricted to $\STCOMP$ admits a final coalgebra.
\end{theorem}

The lower and the compact Vietoris functors on $\TOP$ are seemingly
unrelated, notwithstanding, these functors are closely related when
restricted, respectively, to $\STCOMP$ and $\COMPHAUS$.
From the description of the lower Vietoris functor $\Vie $ on $\POSCH$ we
obtain that the compact Vietoris functor $\Vie :\COMPHAUS\to\COMPHAUS$ is
the composite
\[
  \COMPHAUS\xrightarrow{\text{discrete}}
  \POSCH\xrightarrow{\Vie}\POSCH\xrightarrow{\text{forgetful}}\COMPHAUS.
\]
Being right adjoint, the functor
$\POSCH\xrightarrow{\text{forgetful}}\COMPHAUS$ preserves limits, but
also the inclusion functor $\COMPHAUS\to\POSCH$ does so. As an
interesting consequence, studying preservation of limits by the lower
Vietoris functor in $\STCOMP\simeq\POSCH$ encompasses studying
preservation of limits by the compact Vietoris in $\COMPHAUS$. In
particular, the following results come for free.
 
\begin{corollary}
  The compact Vietoris $\Vie :\COMPHAUS\to\COMPHAUS$ preserves codirected
  limits.
\end{corollary}

\begin{corollary}
  \label{cor:theo_codl_compact}
  All compact Vietoris polynomial functors $F:\COMPHAUS\to\COMPHAUS$
  preserve codirected limits.
\end{corollary}

By taking advantage of the fact that a compact subspace of an
Hausdorff space is a compact Hausdorff space, \cite{Zen70} proves this
property of the compact Vietoris functor even for Hausdorff spaces.

\begin{theorem}
  The compact Vietoris functor $\Vie :\HAUS\to\HAUS$ preserves codirected
  limits.
\end{theorem}

The following results then emerge in a straightforward manner.

\begin{theorem}
  \label{theo_codl}
  All compact Vietoris polynomial functors $F:\HAUS\to\HAUS$ preserve
  codirected limits.
\end{theorem}

\begin{proof}
  Follows from the previous theorem and the fact that all polynomial
  functors $F : \HAUS \to \HAUS$ preserve codirected limits.
\end{proof}

\begin{corollary}
  Let $F : \HAUS \to \HAUS$ be a compact Vietoris polynomial
  functor. The associated category of coalgebras $\Coalg(F)$ is
  complete.
\end{corollary}
\begin{proof}
  Being an epireflective subcategory of $\TOP$, the category $\HAUS$
  is complete and cocomplete, and regularly wellpowered. Furthermore,
  $\HAUS$ is (Epi,RegMono)-structured; but note that $f:X\to Y$ in
  $\HAUS$ is a regular monomorphism if and only if $f$ is a closed
  embedding. It is straightforward to prove that the compact Vietoris
  functor preserves closed embeddings; therefore, by
  Theorem~\ref{theo_hughes}, $\Coalg(F)$ has equalisers. As an
  alternative, $\HAUS$ is also (Surjection, Embedding)-structured; and
  now use Corollary~\ref{cor:general-completeness-coalgebras} and
  Corollary~\ref{cor:Vietoris-pol-preserves-mono-cones} to conclude
  that $\Coalg(F)$ has equalisers. Then the assertion follows from
  Theorem~\ref{theo_codl} and \cite[Theorem 2.1]{Bar93}.
\end{proof}

\begin{theorem}
  Let $F : \TOP \to \TOP$ be a Vietoris polynomial functor that can be
  restricted to $\HAUS$. Then, the category $\Coalg(F)$ has a final
  coalgebra.
\end{theorem}

\begin{proof}
  A consequence of the fact that $I:\HAUS\to\TOP$ preserves and
  reflects limits (\cf\ \cite{AHS90}).
\end{proof}

To close this section we will relate its results with the works
\cite{Kupke:2004,BKR07}. Recall that the former considers compact
Vietoris polynomial functors over $\STONE$. The latter consider
coalgebras for the lower Vietoris functor in the category $\SPECTRAL$
of spectral spaces and spectral maps.

The categories $\STONE$ and $\SPECTRAL$ have a close relation with
some of the categories we considered so far, in particular $\COMPHAUS$
and $\STCOMP$. By taking advantage of this relation we will see that
the fact that every compact Vietoris functor $F : \STONE \to \STONE$
admits a final coalgebra (as shown in \cite{Kupke:2004}) is actually a
consequence of Corollary~\ref{cor:theo_codl_compact}, and the fact
that every lower Vietoris polynomial functor
$F: \SPECTRAL \to \SPECTRAL$ admits a final coalgebra is a direct
consequence of Theorem \ref{theo_lower_comp}.

\begin{remark}
  \label{rem_stone}
  Recall that a Stone space $X$ is a compact Hausdorff space with a
  basis of clopen sets. This is equivalent to saying that $X$ is
  compact Hausdorff and that the cone of continuous maps $(X \to 2)$
  to the discrete two-point-space is initial.
\end{remark}

\begin{lemma}
  Let $(X \to X_i)_{i \in I}$ be a initial cone in $\COMPHAUS$ where
  $X_i$ is a Stone space for every $i \in I$. Then $X$ is a Stone
  space as well.
\end{lemma}

\begin{proof}
  Follows from the fact that each space $X_i$ defines a initial cone
  of continuous maps $(X_i \to 2)$ and that initial cones are closed
  under composition.
\end{proof}

\begin{corollary}
  The canonical forgetful functor $\STONE \to \COMPHAUS$ creates
  limits. Hence, the category $\STONE$ is complete, and the functor
  $\STONE \to \COMPHAUS$ preserves and reflects limits.
\end{corollary}

\begin{theorem}
  Every compact Vietoris polynomial functor $F : \STONE \to \STONE$
  preserves codirected limits.
\end{theorem}

\begin{proof}
  Observe that every compact Vietoris polynomial functor
  $F : \STONE \to \STONE$ is also a functor
  $F: \COMPHAUS \to \COMPHAUS$ and that the diagram below
  commutes. The claim then follows directly from the fact that the
  functor $\STONE \to \COMPHAUS$ preserves and reflects limits.
  \[
    \xymatrix{
      \STONE \ar[d] \ar[r]^{F} & \STONE  \ar[d] \\
      \COMPHAUS \ar[r]_{F} & \COMPHAUS }
  \]
\end{proof}

\begin{corollary}
  Every compact Vietoris polynomial functor $F : \STONE \to \STONE$
  admits a final coalgebra.
\end{corollary}

Analogous results can be achieved for the category $\SPECTRAL$. To see
this let us start a remark akin to Remark \ref{rem_stone}.

\begin{remark}
  Recall that a spectral space $X$ is a stably compact space with a
  basis of compact open subsets. This is equivalent to saying that $X$
  is stably compact and that the cone of spectral maps $(X \to 2)$ to
  the Sierpi{\' n}ski space is initial.
\end{remark}

\begin{lemma}
  Let $(X \to X_i)_{i \in I}$ be a initial cone in $\STCOMP$ where
  $X_i$ is a spectral space for every $i \in I$. Then $X$ is a
  spectral space as well.
\end{lemma}

\begin{proof}
  Follows from the fact that each space $X_i$ defines a initial cone
  of continuous maps $(X_i \to 2)$ to the Sierpi{\' n}ski space and
  that initial cones are closed under composition.
\end{proof}

\begin{corollary}
  The canonical forgetful functor $\SPECTRAL \to \STCOMP$ creates
  limits. Hence, the category $\SPECTRAL$ is complete, and the functor
  $\SPECTRAL \to \STCOMP$ preserves and reflects limits.
\end{corollary}

\begin{theorem}
  Every lower Vietoris polynomial functor
  $F : \SPECTRAL \to \SPECTRAL$ preserves codirected limits.
\end{theorem}

\begin{proof}
  Observe that every lower Vietoris polynomial functor
  $F : \SPECTRAL \to \SPECTRAL$ is also a functor
  $F: \STCOMP \to \STCOMP$ and that the diagram below commutes. The
  claim then follows directly from the fact that the functor
  $\SPECTRAL \to \STCOMP$ preserves and reflects limits.
  \[
    \xymatrix{
      \SPECTRAL \ar[d] \ar[r]^{F} & \SPECTRAL  \ar[d] \\
      \STCOMP \ar[r]_{F} & \STCOMP }
  \]
\end{proof}

\begin{corollary}
  Every lower Vietoris polynomial functor
  $F : \SPECTRAL \to \SPECTRAL$ admits a final coalgebra.
\end{corollary}

\section{Limits via adjunction}
\label{Sec_morelimits}

In this section we extend the results of the previous section to
subfunctors of (Vietoris) polynomial functors, by making use of
adjunction. To achieve this we introduce a number of conditions which
guarantee that a functor $\Coalg(F)\to\Coalg(G)$ induced by a natural
transformation $F\to G$ has a right adjoint: note that if the functor
$\Coalg(F)\to\Coalg(G)$ is also fully faithful, then we can easily
show that $\Coalg(F)$ is ``as complete as'' $\Coalg(G)$. A key
property we use here is a straightforward generalisation of the notion
of taut natural transformation originally introduced in \cite{Mob83}
and \cite{Man02}.

We start with the definition below.

\begin{definition}
  Every natural transformation $\sigma : F \to G$ induces a functor
  $I : \Coalg(F) \to \Coalg(G)$, defined by
  \begin{flalign*}
    & I (X,c) = (X, \sigma_X \comp c), \hspace{0.5cm} I f = f.
  \end{flalign*}
\end{definition}

\noindent
Note that the functor $I : \Coalg(F) \to \Coalg(G)$ is
faithful. Moreover,

\begin{proposition}
  \label{prop:mono_inc}
  If $\sigma : F \to G$ is a monomorphic natural transformation, then
  the functor $I : \Coalg(F) \to \Coalg(G)$ is also full.
\end{proposition}

\begin{proof}
  Take a homomorphism $ f : I(X,c) \to I(Y,d)$. By assumption, the
  equation $G f \comp \sigma_X \comp c = \sigma_Y \comp d \comp f$
  holds. Then, use naturality and the fact that
  $\sigma_Y : FY \to G Y$ is a monomorphism to show that
  $F f \comp c = d \comp f$.
\end{proof}

\noindent
We will now show that, under some assumptions on the natural
transformation $\sigma : F \to G$, the functor above has a right
adjoint.

\begin{assumption}\label{ass:for-I-beeing-adjoint}
  In the remainder of this section the letter $\catC$ denotes a
  category with an $(E,M)$-factorisation structure where $M$ is
  included in the class of monomorphisms. We assume that $\catC$ is
  $M$-wellpowered, that $\sigma : F \to G$ is a natural transformation
  between endofunctors on $\catC$ where every component $\sigma_X$ is
  in $M$, and that $G$ sends morphisms in $M$ to morphisms in $M$.
\end{assumption}

\begin{theorem}\label{thm:adjoint_from_nat_transformation}
  Under Assumption~\ref{ass:for-I-beeing-adjoint} with $\catC$
  cocomplete, the functor $I : \Coalg(F) \to \Coalg(G)$ is left
  adjoint.
\end{theorem}
\begin{proof}
  We will show that the assumptions of the General Adjoint Functor
  Theorem hold. Since $\catC$ is cocomplete, the category $\Coalg(F)$
  is cocomplete as well. Moreover, $I : \Coalg(F) \to \Coalg(G)$
  preserves colimits, as $U I : \Coalg(F) \to \catC$ preserves
  colimits, and the forgetful functor $U : \Coalg(G) \to \catC$
  reflects them. It remains to verify the Solution Set Condition.  For
  this, take a coalgebra $d:Y\to GY$. Let $\mathcal{S}_0$ be a set of
  representatives of the collection of all $\catC$-objects $Q$
  admitting an $M$-morphism $Q\to Y$, and let $\mathcal{S}$ be the set
  of all $F$-coalgebras based on an object in $\mathcal{S}_0$. Let now
  $(X,c)$ be an $F$-coalgebra and $f:(X,\sigma_X\comp c)\to(Y,d)$ be a
  homomorphism of $G$-coalgebras. By hypothesis, $f:X\to Y$ factorises
  as $f=m\comp e$
  \[
    X\xrightarrow{\;e\;}Q\xrightarrow{\;m\;}Y
  \]
  with $e\in E$ and $m\in M$. Since $\sigma_Q:FQ\to GQ$ and
  $Gm:GQ\to GY$ are in $M$, there is a diagonal $q:Q\to FQ$ so that
  the right hand square and the lower-left square in
  \[
    \xymatrix{GX\ar[r]^{Ge} & GQ\ar[r]^{Gm} & GY\\
      FX\ar[u]^{\sigma_X}\ar[r]_{Fe} & FQ\ar[u]_{\sigma_Q}\\
      X\ar[u]^c\ar[r]_e & Q\ar[r]_m\ar@{..>}[u]_q & Y\ar[uu]_d}
  \]
  commute; the upper-left square commutes since $\sigma$ is a natural
  transformation. This proves that $f:(X,\sigma_X\comp c)\to(Y,d)$
  factorises via the image of an object in $\mathcal{S}$.
\end{proof}

\begin{corollary}
  \label{cor_closed}
  The category $\Coalg(F)$ has all (co)limits of a certain type if
  $\Coalg(G)$ does so.
\end{corollary}

\begin{corollary}
  Let $F:\TOP\to\TOP$ be a compact Vietoris polynomial functor that
  can be restricted to $\HAUS$. Every subfunctor of $F$ admits a final
  coalgebra.
\end{corollary}

\begin{remark}
  The Corollary above applies to various interesting variants of the
  compact Vietoris functor that were not yet mentioned. In particular,
  \begin{itemize}
  \item the one that discards the empty set,
  \item analogously to the finitary powerset functor, the one that
    takes infinite sets out of comission, and
  \item the one which considers only compact and connected subsets
    (\cf\ \cite{Dud72}).
  \end{itemize}
  All these variants are subfunctors of the compact Vietoris
  functor. In conjunction with the polynomial ones, they form a family
  of subfunctors of compact Vietoris polynomial functors.
\end{remark}

\begin{corollary}
  Let $F:\TOP\to\TOP$ be a lower Vietoris polynomial functor that can
  be restricted to $\STCOMP$. Every subfunctor of $F$ admits a final
  coalgebra.
\end{corollary}

\noindent
The proof of Theorem~\ref{thm:adjoint_from_nat_transformation} gives
us also a hint on how to construct a coreflection of a $G$-coalgebra
$(Y,d)$: take the ``largest $M$-subcoalgebra of $(Y,d)$''. In the
sequel we make this idea more precise. To do so, motivated by
\cite{Mob83} and \cite{Man02}, we introduce the following notion.

\begin{definition}
  A natural transformation $\sigma : F \to G$ is \df{$M$-taut} if each
  naturality square induced by a morphism in $M$ is a pullback square;
  that is, for every morphism $m:X\to Y$ in $M$, the diagram below is
  a pullback square.
  \[
    \xymatrix{
      F X \ar[r]^{F m} \ar[d]_{\sigma_X} & F Y \ar[d]^{\sigma_Y} \\
      G X \ar[r]_{G m} & G Y }
  \]
\end{definition}

\noindent
Recall from \cite[Definition 7.79]{AHS90} that, for monomorphisms
$m_1:M_1\to X$ and $m_2:M_2\to X$ in a category, $m_1$ is \df{smaller
  than} $m_2$ (written as $m_1\le m_2$) whenever there is some
$m:M_1\to M_2$ with $m_2\comp m=m_1$. Note that $m$ is necessarily a
monomorphism. Assuming that $\catC$ has pullbacks, take a
$G$-coalgebra $(Y,d)$ and consider the pullback square
\begin{equation}\label{eq:pullback_S}
  \xymatrix{S\ar[d]_i\ar[r] & FY\ar[d]^{\sigma_Y}\\ Y\ar[r]_d & GY}
\end{equation}
in $\catC$. Note that $i:S\to Y$ is in $M$, by \cite[Proposition
14.15]{AHS90}.

\begin{lemma}
  \begin{enumerate}
  \item For every $F$-coalgebra $(X,c)$ and every homomorphism
    $m:(X,\sigma_X\comp c)\to(Y,d)$ with $m\in M$, $m$ is smaller than
    $i:S\to Y$.
  \item Assume now that the natural transformation $\sigma:F\to G$ is
    $M$-taut and let $m:(Q,q)\to(Y,d)$ be a homomorphism in
    $\Coalg(G)$ where $m\in M$ and $m\le i$. Then there is a
    $F$-coalgebra structure $q':Q\to FQ$ on $Q$ with
    $\sigma_Q\comp q'=q$.
  \end{enumerate}
\end{lemma}
\begin{proof}
  An easy calculation, and \cite[Proposition 14.9]{AHS90} show that
  the first claim is true. In regard to the second one, let
  $\bar{m}:Q\to S$ be the arrow in $\catC$ with
  $i\comp\bar{m}=m$. Then, since in the diagram
  \[
    \xymatrix{& Y\ar[rr]^d && GY\\
      S\ar[ru]_i\ar[rr] && FY\ar[ru]_{\sigma_Y}\\
      &&& GQ\ar[uu]_{Gm}\\
      Q \ar[uu]^{\bar{m}} \ar@/_{5em}/[rrru]_c \ar@/^{6em}/[uuur]^m
      \ar@{..>}[rr]
      && FQ\ar[uu]^{Fm}\ar[ur]^{\sigma_Q}\\
      & }
  \]
  the right hand parallelogram is a pullback square and the outer
  diagram and the top parallelogram commute. This provides the desired
  arrow $Q\to FQ$.
\end{proof}

\noindent
For a $G$-coalgebra $(Y,d)$, the class of all subcoalgebras
$m:(X,c)\to(Y,d)$ with $m\in M$ is preordered under the smaller-than
relation. Since $\catC$ is $M$-wellpowered, this class is equivalent
to an ordered set; and by a slight abuse of language we will speak of
the ordered set of $M$-subobjects of $(Y,d)$.

Recall from Proposition~\ref{prop:mono_inc} that the induced functor
$I:\Coalg(F)\to\Coalg(G)$ is fully faithful since $\sigma:F\to G$ is a
monomorphic natural transformation. Hence, we can consider $\Coalg(F)$
as a full subcategory of $\Coalg(G)$. From the results above we
obtain:

\begin{theorem}
  In addition to Assumption~\ref{ass:for-I-beeing-adjoint}, assume
  that $\catC$ has pullbacks, $\sigma$ is $M$-taut and, for every
  $G$-coalgebra $(Y,d)$, the ordered set of $M$-subobjects is
  complete. Then, for every $G$-coalgebra $(Y,d)$, the coreflection of
  $(Y,d)$ is given by the supremum $(\bar{Q},\bar{q})\to(Y,d)$ of all
  $G$-homomorphisms $(Q,q)\to(Y,d)$ with $(Q,q)$ in $\Coalg(F)$ and
  $m:Q\to Y$ in $M$ smaller than $i:S\to Y$ (defined by the pullback
  square \eqref{eq:pullback_S}).
\end{theorem}

\begin{remark}
  If $\catC$ has coproducts, then also $\Coalg(G)$ has coproducts
  which guarantees completeness of the ordered set of $M$-subobjects
  of $(Y,d)$. In fact, let $(m_i:(X_i,c_i)\to(Y,d))_{i\in I}$ a family
  of subcoalgebras of $(Y,d)$ with $m_i\in M$, for every $i\in
  I$. Then the supremum $m:(X,c)\to(Y,d)$ of this family is given by
  the $(E,M)$-factorisation of the canonical map
  $f:\coprod_{i\in I}(X_i,c_i)\to (Y,d)$ induced by this family.
  \[
    \xymatrix{\coprod_{i\in I}(X_i,c_i)\ar[r]^-{e}\ar@/^{2em}/[rr]^f & (X,c)\ar[r]^m & (Y,d)\\
      (X,c_i)\ar[u]^{k_i}\ar[urr]_{m_i}}
  \]
\end{remark}

\section{Vietoris coalgebras at work}
\label{Sec_work}
\noindent
Moving to the more practical side, recall the bouncing ball system
mentioned in the introduction.  Formally, it consists of a ball that
is dropped at a certain height $(p)$, and with an initial velocity
$(v)$. Due to the gravitional effect $(g)$, it falls into the ground
and then bounces back up, losing, for example, half of its kinetic
energy. As the documents \cite{neves15,neves16} show, such a behaviour
can be described coalgebraically, with the help of the functor defined
below.

\begin{definition}
  Let $\Rz$ denote the topological space $\mathbb{R}_{\geq 0}$.  Then
  define $\MH : \TOP \rightarrow \TOP$ as the functor such that for
  any topological space $X$, and any continuous map $g : X \to Y$,
  \begin{flalign*}
    & \MH X = \{ (f,d) \in X^{\Rz} \times \Dur \mid f \comp
    \curlywedge_d = f \}, \hspace{1cm} \MH g = g^{\Rz} \times \id
  \end{flalign*}
  \noindent
  where $\Dur$ is the one-point compactification of $\Rz$ and
  $\curlywedge_d = \min(\_ \; , d)$.
\end{definition}

\noindent
Intuitively, the functor $\MH : \TOP \to \TOP$ captures continuous
behaviour as considered in hybrid systems, \ie\ the continuous evolutions
of physical processes, such as the movement of a plane, or the
temperature of a room.  Document \cite{neves16} provides the following
specification for the bouncing ball described above.

\begin{definition}
  Use $S,O$ as shorthand to $\mathbb{R}_{\geq 0} \times \mathbb{R}$,
  and $\mathbb{R}$, respectively. The bouncing ball is given by the
  $\SET$-coalgebra $\pv{\nxt,\out} : S \to S \times U \MH O$
  \begin{flalign*}
    \pv{\nxt,\out} \> (p,v) = \big ( (0,u), (\mov(p,v, \; \_ \;),d)
    \big )
  \end{flalign*}

  \noindent
  where variable $u$ corresponds to the (abrupt) change of velocity
  due to the collision with the ground, function
  $\mov (p,v, \; \_ \; ) : \Rz \rightarrow O$ describes the ball's
  movement between jumps, and $d$ denotes the time that the ball takes
  to reach the ground. In symbols,
  \begin{flalign*}
    u = (v + g d) \times -0.5, \> \> \> \mov (p,v,t) = p + vt +
    \tfrac{1}{2} g t^2, \> \> \> d = \tfrac{\sqrt{2gp + v^2} + v}{g}.
  \end{flalign*}
\end{definition}

\noindent
Recall that for each set $A$ the functor
$(\; \_ \; \times A) : \SET \to \SET$ has a final coalgebra (\cf\
\cite{rutten2000}), thus providing a notion of behaviour for the
ball. To be more concrete, the coalgebra
$\big (S, \pv{\nxt,\out} \big )$ has a canonical homomorphism
$\lana \; \_ \; \rana : S \to (U \MH O)^\omega$ to the final coalgebra
$\big ((U \MH O)^\omega, \pv{\tl,\hd} \big )$, where
$\tl : (U \MH O)^\omega \to (U \MH O)^\omega$, and
$\hd : (U \MH O)^\omega \to U \MH O$ are the `tail' and `head'
functions, respectively. The map $\pv{\tl,\hd}$ computes the behaviour
of the ball for a given height and velocity. For example, the first
three elements of $\lana (0,5) \rana$ yield the following plots.
\begin{multicols}{3}
    
  \scalebox{0.50}{
    \includegraphics{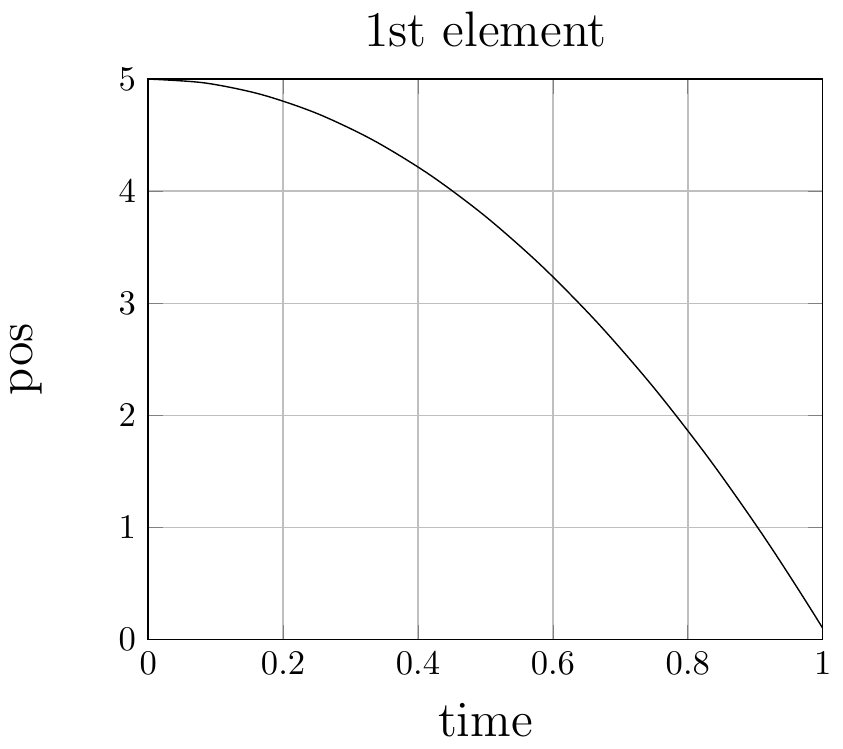}
  } \columnbreak

  \scalebox{0.50}{
    \includegraphics{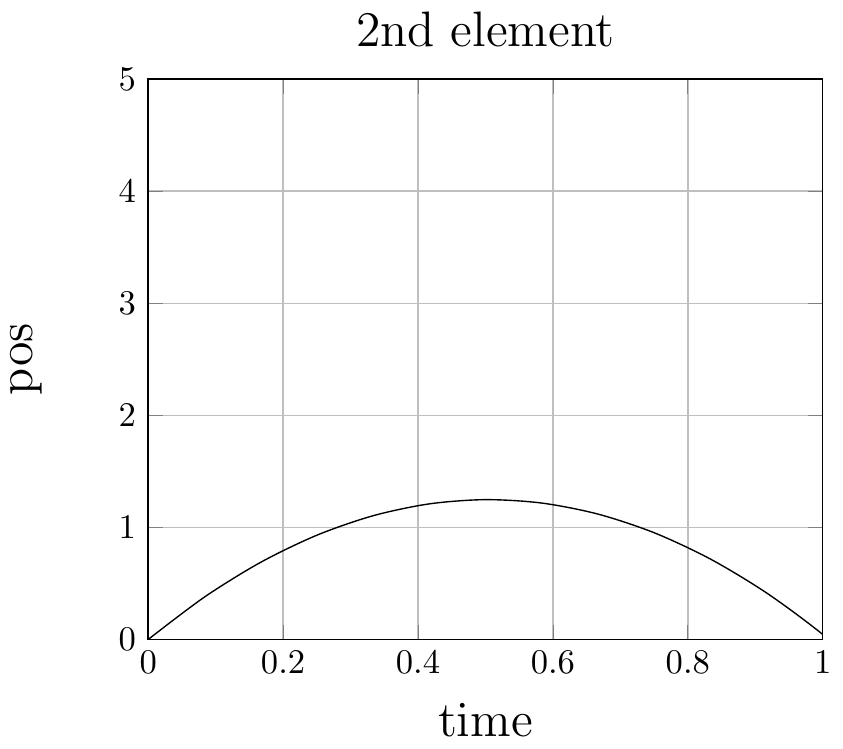}
  } \columnbreak

  \scalebox{0.50}{
    \includegraphics{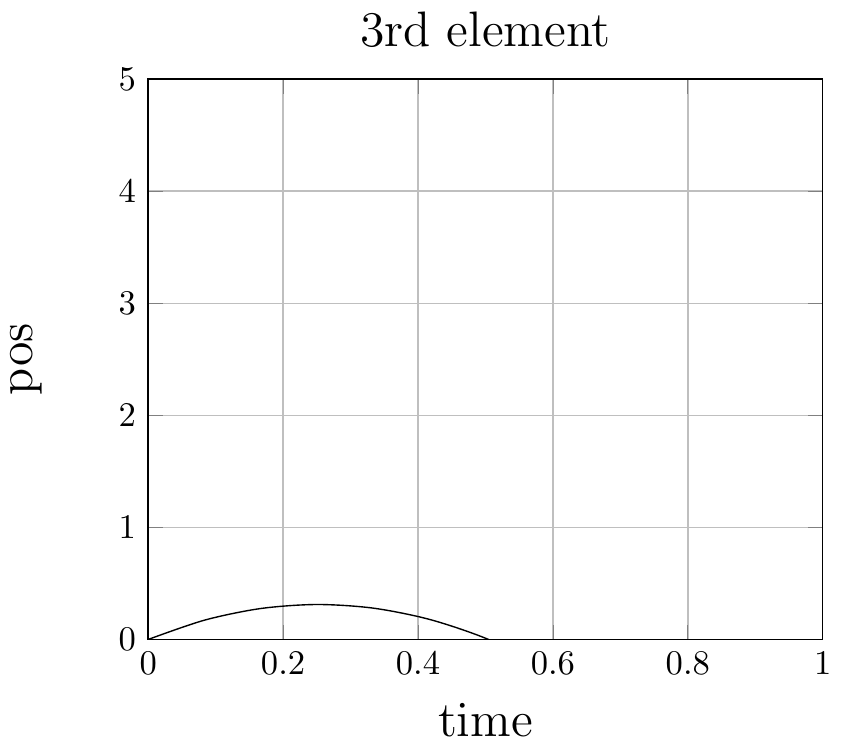}
  }

\end{multicols}

\noindent
In order to bring non-determinism into the scene, suppose, for
example, that when the ball hits the ground it loses part of its
kinetic energy non-deterministically. In this context, one may
consider the coalgebra
$\pv{\nxt,\out} : S \to \mathscr{P} S \times U \MH O$
\begin{flalign*}
  \pv{\nxt,\out} \> (p,v) = \big (U, (\mov(p,v, \; \_ \;),d) \big )
\end{flalign*}
\noindent
with
$U = \big \{ \> \big ( 0,(v + gd) \times c \big ) \in S \mid c \in
[-0.7,-0.5] \> \big \}$.  However, the functor
$(\mathscr{P} \times U \MH O) : \SET \to \SET$ has no final coalgebra
(\cf\ \cite{rutten2000}), and thus there is no canonical notion of
behaviour for the non-deterministic bouncing ball specified above.  We
will show that the issue can be fixed by shifting to $\TOP$. For this,
the following result is useful.

\begin{proposition}
  Let $\Vie : \TOP \to \TOP$ be the compact Vietoris functor.  The
  family $\tau=(\tau_{X,Y})$ of maps
  \begin{align*}
    \tau_{X,Y}:(\Vie X)\times Y & \to \Vie (X\times Y)\\
    (S,y) &\mapsto S\times\{y\}
  \end{align*}
  defines a natural transformation
  \[
    \xymatrix{\TOP\times\TOP\ar[r]^\times\ar[d]_{\Vie \times\Id} &
      \TOP\times\TOP\ar[d]^\Vie\\
      \TOP\times\TOP\ar[r]_-{\times} & \TOP.\ultwocell<\omit>{\tau}}
  \]
\end{proposition}
  
\begin{proof}
  Let $X$ and $Y$ be topological spaces. For all $S\in \Vie X$ and
  $y\in Y$, since $S$ is compact, the product $S \times \{ y \}$ is
  also compact, which entails that
  $S \times \{ y \} \in \Vie (X \times Y)$. Then, continuity of the
  map $\tau_{X,Y}$ is a direct consequence of the equalities below.
  \begin{flalign*}
    & \tau^{-1}_{X,Y} \Big [ \Big ( \bigcup_{i \in I} U_i \times V_i
    \Big )^{\Diamond} \; \Big ]
    = \bigcup_{i \in I} (U_i)^{\Diamond} \times V_i  \\
    & \tau^{-1}_{X,Y} \Big [ \Big ( \bigcup_{i \in I} U_i \times V_i
    \Big )^{\Box} \; \Big ] = \bigcup \Big \{ \Big ( \bigcup_{i \in F}
    U_i \Big )^\Box \times \bigcap_{i \in F} V_i \mid F \subseteq I
    \text{ finite } \Big \}
  \end{flalign*}

  \noindent
  The proof that all naturality squares commute is straightforward.
\end{proof}

\begin{remark}
  When the compact Vietoris functor is equipped with the natural
  transformation above it becomes a strong functor.  The latter
  concept was introduced in \cite{kock:1972} and is widely adopted in
  monadic programming.
\end{remark}

\noindent
With the natural transformation above, it becomes straightforwad to
consider the non-deterministic bouncing ball in a topological
setting. Actually, it can be shown to be a coalgebra
\begin{flalign*}
  \pv{\nxt,\out} : S \to \Vie S \times \MH O
\end{flalign*}

\noindent
First, the map $\out : S \to \MH O$ was already shown to be continuous
in \cite{neves15}. Then, observe that the map $\nxt : S \to \Vie S$
can be rewritten as a composite
\[
  \xymatrix{ S \ar[r]^(0.35){ \langle f,g \rangle } & \Vie S \times
    S^S \ar[r]^{\tau} & \Vie (S \times S^S) \ar[r]^(0.65){\Vie \ev} &
    \Vie S }
\]
\begin{flalign*}
  f \left (p,v \right ) & = \{ 0 \} \times [0.5, 0.7] \\
  g \left (p,v \right ) & = \lambda (x,y) \in S . \left ( 0 , (v + gd)
    \times - y \right )
\end{flalign*}
\noindent
which proves our claim. One more result is needed.

\begin{theorem}
  \label{H_HAUS}
  The functor $\MH : \TOP \to \TOP$ can be restricted to the category
  of Hausdorff spaces.
\end{theorem}

\begin{proof}
  Let $X$ be a locally compact space and $Y$ an Hausdorff space.
  Then, the function space $Y^X$ equipped with the compact-open
  topology is Hausdorff (\cf\ \cite{Kelley}). The claim now follows
  from Hausdorff spaces being closed under products, and subspaces.
\end{proof}

\noindent
As discussed in the previous sections, every compact Vietoris
polynomial functor that can be restricted to the category of Hausdorff
spaces has a final coalgebra, which, according to Theorem
\ref{H_HAUS}, is the case for $\Vie \times \MH O : \TOP \to \TOP$.
Intuitively, the elements of the final $(\Vie \times \MH O)$-coalgebra
can be seen as compactly branching trees, \ie\ trees where the set of
sons of each node is compact.  This is similar to the property imposed
to finitely branching trees, which occur in the final coalgebras
involving the \emph{finite} powerset functor (\cf\ \cite{rutten2000}).
Interestingly, the functor $\Vie \times \MH O : \TOP \to \TOP$ admits
an alternative representation: superimpose the evolutions of each
level of the tree. To illustrate this, the non-deterministic bouncing
ball yields the following plots for the first two bounces, with the
pair $(5,0)$ as the initial state.
\begin{multicols}{3}
    
  \scalebox{0.50}{
    \includegraphics{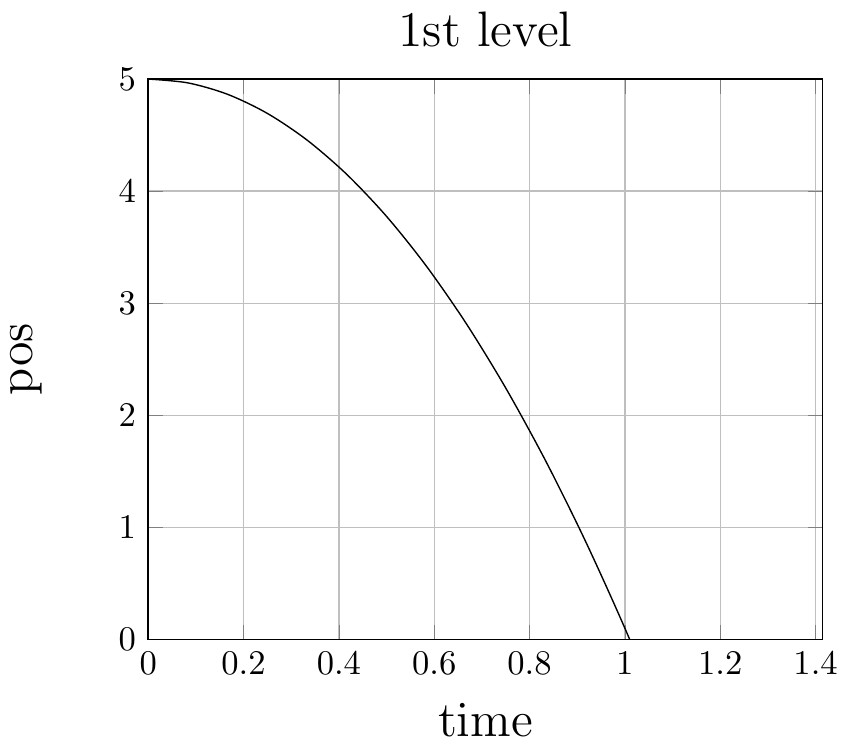}
  } \columnbreak

  \scalebox{0.50}{
    \includegraphics{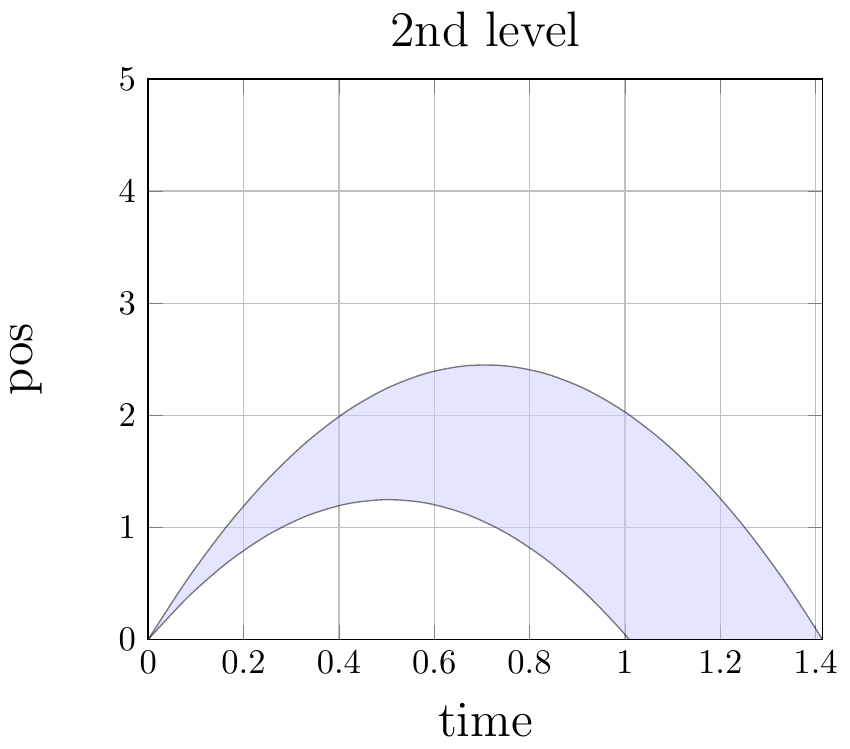}
  } \columnbreak

  \scalebox{0.50}{
    \includegraphics{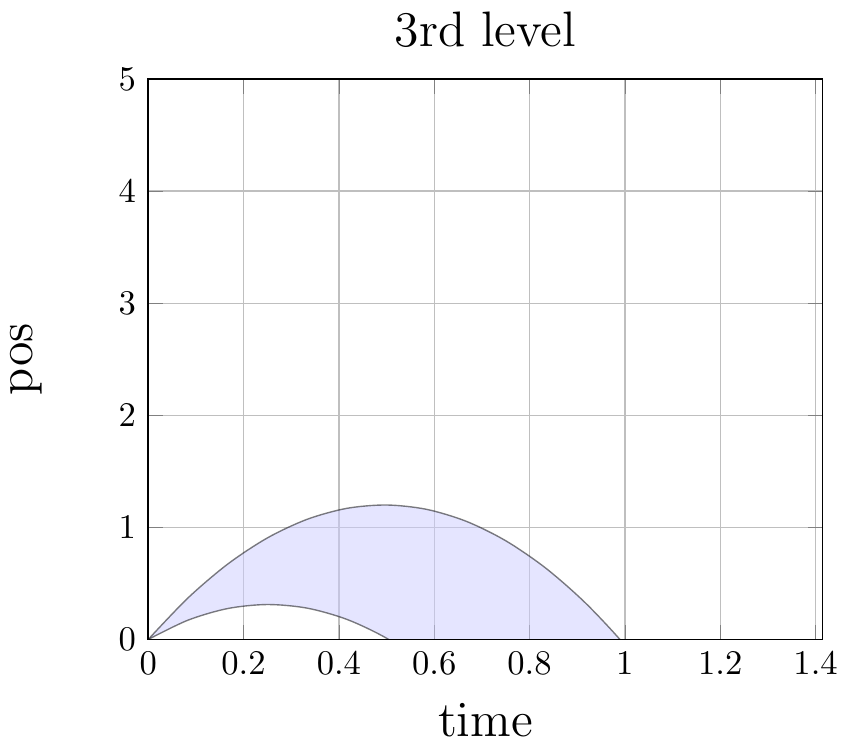}
  }

\end{multicols}

\noindent
The notion of \df{stability} \cite{Stauner_01} is another important
aspect in the development of hybrid systems. Roughly put, the term
`stability' refers to a system's stability in regard to its behaviour
against perturbations; the system is called stable if small changes in
its state (or input) only produce small changes in its behaviour ---
such a notion is directly related to that of distance between
behaviours, which was already studied in a coalgebraic setting
\cite{BaldanBKK14}.

In a $\SET$-based context it is difficult to reason about the
stability of a system, because its state space, which is assumed to be
just a set, lacks sufficient structure. In the topological setting,
however, the issue can be better handled.  To start with, observe that
topological spaces already carry a notion of proximity, given by the
open sets. Moreover, note that the notion of a stable system is
closely related to that of a continuous map, as discussed, for
example, in \cite{Stauner_01}. This relation can be precisely
described in a coalgebraic context: take a functor
$F : \TOP \to \TOP$, and assume that $\Coalg(F)$ has a final coalgebra
$(\nu_F, \omega_F)$. Then, for any $F$-coalgebra $(S,c)$ there is a
continuous map $\lana \; \_ \; \rana : S \to \nu_F$ such that for each
state $s \in S$, $\lana s \rana$ is the associated behaviour. Since
the map is continuous, `close' states must have `close' behaviours,
which coincides with our notion of system stability.  This suggests
the following coalgebraic definition of stability.

\begin{definition}
  \label{defn_bb}
  Let $F : \TOP \to \TOP$ be a functor that admits a final
  coalgebra. Then a (not necessarily continuous) map $c : X \to F X$
  is called stable if it is a member of $\Coalg(F)$. In other words,
  if it is a continuous map.
\end{definition}

\begin{examples}
  The bouncing balls $\pv{\nxt,\out} : S \to S \times \MH O$, and
  $\pv{\nxt,\out} : S \to \Vie S \times \MH O$ are continuous maps,
  and, consequently, stable systems. In this case calling either of
  the bouncing balls stable, is to say that a small change in their
  initial position and velocity does not drastically alter their
  (possible) trajectories over time.
\end{examples}

\noindent
Finally, note that the systems considered here jump between states
discretely, as opposed to their outputs which are, essentially,
evolutions in time of specific values. One possible way to accomodate
the evolution of states as well is to consider coalgebras in
$\Coalg(\MH)$.  We will use the results of the previous sections to
show that this category is also complete.

\begin{definition}
  Let $F : \catC \to \catC$ be a functor over a category $\catC$ with
  (co)products. We call $F$ \df{exponent polynomial} if it can be
  recursively defined from the grammar below, with letters $A$ and $B$
  denoting, respectively, an arbitrary object
  and an exponentiable
  object of $\catC$.
  \begin{flalign*}
    F ::= F + F \mid F \times F \mid A \mid \Id \mid F^B
  \end{flalign*}
\end{definition}

\noindent
Since all exponential functors $(\> \_ \>)^B : \catC \to \catC$ are
right adjoints, the following results come almost for free.

\begin{proposition}
  All exponent polynomial functors $F : \TOP \to \TOP$ preserve
  connected limits.
\end{proposition}

\begin{corollary}
  The categories of coalgebras of all exponent polynomial functors
  over $\TOP$ are complete.
\end{corollary}

\begin{theorem}
  The category of coalgebras $\Coalg(\MH)$ is complete.
\end{theorem}

\begin{proof}
  The previous corollary assures that the category
  $\Coalg \left ((\> \_ \>)^\Rz \times \Dur \right )$ is complete.
  Then, observe that the functor $\MH : \TOP \to \TOP$ is a subfunctor
  of $(\> \_ \>)^\Rz \times \Dur : \TOP \to \TOP$, and apply Theorem
  \ref{thm:adjoint_from_nat_transformation}.
\end{proof}

The previous theorem takes advantage of the adjoint situation below.
\begin{flalign*}
  \Coalg \left ((\> \_ \> )^\Rz \times \Dur \right ) \adjunct{}{}
  \Coalg \left (\MH \right )
\end{flalign*}

Then, with the theorem below, and using the results of the previous
sections, we obtain a specific method to construct coreflections of
$((\> \_ \>)^\Rz \times \Dur)$-coalgebras.

\begin{theorem}
  The `inclusion' natural transformation
  $\iota : \MH \to (\> \_ \>)^\Rz \times \Dur$ is mono-taut.
\end{theorem}

\begin{proof}
  Consider a monomorhism $m : X \to Y$ in $\TOP$. We will show that
  the diagram below is a pullback square.
  \[
    \xymatrix{
      \MH X \ar[r]^{\MH m} \ar[d]_{\iota_X} & \MH Y \ar[d]^{\iota_Y} \\
      X^{\Rz} \times \Dur \ar[r]_{m^{\Rz} \times id} & Y^{\Rz} \times
      \Dur }
  \]

    \noindent
    Thus, take two morphisms $f : Z \to X^\Rz \times \Dur$,
    $g : Z \to \MH Y$, and assume that the equation below holds.
    \begin{flalign*}
      (m^\Rz \times id) \comp f = \iota_Y \comp g
    \end{flalign*}

    \noindent
    Let $z \in Z$ and put $(x,y)=f(z)$ and $(a,b)=g(z)$. Then, by the
    definition of $\MH$, $a = a \comp \curlywedge_b$ since
    $ \im g \subseteq \MH Y$. Using
    $(m^\Rz \times id) \comp f = \iota_Y \comp g$, one gets
    $ m \comp x = m \comp x \comp \curlywedge_y$; and from this, one
    obtains $x = x \comp \curlywedge_y$ since $m :X \to Y$ is a
    monomorphism.  This shows that the condition
    $\im f \subseteq \MH X$ holds.  Then, since the map
    $\iota_X : \MH X \to X^\Rz \times \Dur$ is an embedding, and
    $\im f \subseteq \im \iota_X$, there must be a unique arrow
    $h : Z \to \MH X$ such that $\iota_X \comp h = f$.  It remains to
    show that $g = \MH m \comp h$. This is a direct consequence of the
    diagram above being commutative, and the map
    $\iota_Y : \MH Y \to Y^\Rz \times \Dur$ mono.
  \end{proof}

\section{Conclusions and future work}
  \label{Sec_conc}

\noindent
Even if most coalgebraic literature takes $\SET$ as the base category,
state-based transition systems often call for a shift to other
categories, where mechanisms that suitably handle their intricacies
are available. Such was the case in \cite{prakash_2009,Doberkat:2009},
two research lines on the topic of stochastic systems, and in
\cite{Kupke:2004,Venema10,venema2014}, where the category of Stone
spaces and continuous maps plays a key role in setting an appropriate
coalgebraic semantics for finitary modal logics.

In our case the base category adopted was $\TOP$. As discussed in the
previous section, this was because the $\SET$-based context proved to
be insufficient for the design of (non-deterministic) hybrid systems,
namely in what concerns canonical representations of behaviour and
stability. The shift to the topological setting provided, almost for
free, a notion of stability (in the spirit of \cite{Stauner_01}), and
showed that a number of non-deterministic hybrid systems in $\TOP$
have an associated final coalgebra, even if in $\SET$ they do
not. Both results were achieved using this paper's theoretical
developments. But again, we stress that the latter can be applied to
other contexts as well.

The relevance of Vietoris coalgebras for different topics is further
witnessed by the common existence of important limits in their
categories of coalgebras. We saw that every compact Vietoris
polynomial functor admits a final coalgebra if it can be restricted to
the category $\HAUS$ while every lower Vietoris polynomial functor
admits a final coalgebra if it can be restricted to
$\STCOMP$. Moreover, we saw that several variants of such functors
also inherit these results and that all categories of Vietoris
coalgebras have equalisers.

However, several theoretical questions concerning limits in categories
of Vietoris coalgebras still remain open. For example, we studied
codirected limit preservation by Vietoris functors under different
topological contexts (see Section \ref{Sec_limits}), showing cases in
which they were preserved, and cases in which they were not. But we
are still not precisely sure what is the `weakest' context in which
they are preserved. Another example concerns the existence of products
in categories of Vietoris coalgebras. Recall also our study of
topological functors between categories of coalgebras. Among other
things, it provides a full characterisation of situations in which it
is possible to systematically lift well-known results about coalgebras
over $\SET$ to coalgebras over other categories. We saw that this is
indeed the case between coalgebras of polynomial functors over $\SET$
and their counterparts in $\TOP$, but we are also interested in other
situations. Two prime examples that we will explore in future work
pertain coalgebras over the category $\ORD$ and coalgebras over the
category $\PMET$. These coalgebras have significant relevance within
the coalgebraic community (\eg\ \cite{BaldanBKK14,BalanK11,BalanKV13})
and we believe that our study can contribute to the topic.

On a note closer to practice, the use of topologies to specify and
analyse (non-deterministic) hybrid systems brings a number of
benefits, which were just barely grasped in this paper. Our main goal
is to further explore them in the near future. The plan is to do so in
a coalgebraic component-based approach \cite{Barbosa03,HasuoJ11},
where simple hybrid systems can be composed to form more complex
ones. The results that this paper reports provide an interesting step
in this direction.

\section*{Acknowledgements}

We are grateful to Lawrence Moss for many fruitful discussions and for
his helpful references on the subject. This work is financed by the
ERDF -- European Regional Development Fund through the Operational
Programme for Competitiveness and Internationalisation -- COMPETE 2020
Programme and by National Funds through the Portuguese funding agency,
FCT -- Funda\c{c}\~{a}o para a Ci\^{e}ncia e a Tecnologia within
project POCI-01-0145-FEDER-016692. We also gratefully acknowledge
partial financial assistance by Portuguese funds through CIDMA (Center
for Research and Development in Mathematics and Applications), and the
Portuguese Foundation for Science and Technology (``FCT --
Funda\c{c}\~ao para a Ci\^encia e a Tecnologia''), within the project
UID/MAT/04106/2013. Finally, Renato Neves and Pedro Nora are also supported by
FCT grants SFRH/BD/52234/2013 and SFRH/BD/95757/2013, respectively.

\newcommand{\etalchar}[1]{$^{#1}$}


\end{document}